\documentclass[a4paper]{article}
\usepackage[utf8]{inputenc}
\usepackage{amsmath}
\usepackage{amsfonts}
\usepackage{amssymb}
\usepackage{amsthm}
\usepackage{bbm}
\usepackage{xcolor}
\usepackage{graphicx}
\usepackage{hyperref}
\numberwithin{equation}{section}
\usepackage[top=1in,
  bottom=1.2in,
  left=1in,
  right=0.75in]{geometry}

\newtheorem{theorem}{Theorem}[section]
\newtheorem{lemma}[theorem]{Lemma}
\newtheorem{definition}[theorem]{Definition}

\newtheorem{corollary}[theorem]{Corollary}
\newtheorem*{notation*}{Notation}
\newtheorem{remark}[theorem]{Remark}
\newtheorem{example}[theorem]{Example}

\theoremstyle{definition}
\newtheorem{assumption}[theorem]{Assumption}

\newtheorem{manualtheoreminner}{Assumption}
\newenvironment{assum}[1]{%
  \renewcommand\themanualtheoreminner{#1}%
  \manualtheoreminner
}{\endmanualtheoreminner}

\title{Sandwiched Volterra Volatility model: 
\\ Markovian approximations and hedging}

\author{Giulia Di Nunno$^{1,2}$\\\href{mailto:giulian@math.uio.no}{giulian@math.uio.no}
   \and Anton Yurchenko-Tytarenko$^1$ \\ \href{mailto:antony@math.uio.no}{antony@math.uio.no}}

\date{%
    $^1$Department of Mathematics, University of Oslo\\%
    $^2$Department of Business and Management Science, NHH Norwegian School of Economics, Bergen\\[2ex]%
    July 9, 2024
}

\begin{document}

\maketitle

\begin{abstract}
    We consider stochastic volatility dynamics driven by a general H\"older continuous Volterra-type noise and with unbounded drift. For these so-called SVV-models, we consider the explicit computation of quadratic hedging strategies. While the theoretical hedge is well-known in terms of the non-anticipating derivative for all square integrable claims, the fact that these models are typically non-Markovian provides is a challenge in the direct computation of conditional expectations at the core of the explicit hedging strategy. To overcome this difficulty, we propose a Markovian approximation of the model which stems from an adequate approximation of the kernel in the Volterra noise. We study the approximation of the volatility, of the prices and of the optimal mean-square hedge. We provide the corresponding error estimates. The work is completed with numerical simulations.
\end{abstract}

\noindent\textbf{Keywords:} stochastic volatility, sandwiched process, H\"older continuous noise, hedging, Monte Carlo methods\\
\textbf{MSC 2020:} 91G30; 60H10; 60H35; 60G22 \\[9pt]
\textbf{Acknowledgements.} The present research is carried out within the frame and support of the ToppForsk project nr. 274410 of the Research Council of Norway with title STORM: Stochastics for Time-Space Risk Models.

\section{Introduction}\label{s: introduction paper 5}


A significant number of modern financial market models are characterized by the absence of the Markov property. For a detailed discussion of such a choice of modeling framework, we refer the reader to a recent review \cite{Di_Nunno_Kubilius_Mishura_Yurchenko-Tytarenko_2023}; here we only note that the non-Markovian structure is substantiated by several empirical considerations coming from the analysis of financial time series and implied volatility (IV) surfaces.

\begin{itemize}
    \item A number of studies (see e.g. \cite{AndersenBollerslev1997, AndersenBollerslevDieboldLabys2001, BollerslevMikkelsen1996, DingGrangerEngle1993} or \cite{Cont2005, Cont2006}) report the presence of \textit{long memory} in financial markets. Additionally, Comte and Renault \cite{ComteRenault1998} note that the observed IV smile amplitude decreases much slower than predicted by standard models, when the time to maturity $T\to\infty$, which can also be interpreted as a manifestation of volatility persistence. These observations resulted in multiple works utilizing long memory processes such as fractional Brownian motion (fBm) with Hurst index $H>1/2$ for volatility modeling (see e.g. \cite{BMdP2018, ChronopoulouViens2012, ComteCoutinRenault2012} or \cite{Rosenbaum_2008}).

    \item Other studies (such as \cite{Fukasawa_Takabatake_Westphal_2019, GatheralJaissonRosenbaum2018}) estimate the H\"older order of the market volatility to be in the vicinity of 0.1. Additionally, this idea of modeling volatility by a process with low regularity is supported by the behavior of the short-term skew slope of the at-the-money IV \cite{Alos_Leon_Vives_2007, Fukasawa_2021}. These facts gave birth to a number of so-called \textit{rough volatility} models and, while \textit{roughness} in itself does not necessarily yield the absence of Markovianity, many rough volatility models in the literature (see e.g. \cite[Section 7.2]{Alos_Leon_Vives_2007} as well as \cite{BayerFrizGatheral2016, EuchRosenbaum2019, Fukasawa_Gatheral_2022, Harms_Stefanovits_2019}) exploit the structure of fBm with $H<1/2$ and hence are non-Markovian.
\end{itemize}
It should be noted, however, that fBm-based models are far from perfect, despite being typical choices for introducing long memory/roughness.
\begin{itemize}
    \item[(i)] In fractional models, there seems to be an inherent contradiction in the choice of the Hurst index $H$: long memory and the behavior of the IV surfaces for longer maturities yield $H>1/2$ whereas the power law of the IV skew slope and roughness demand $H<1/2$, see e.g. the discussion in \cite{Funahashi_Kijima_2017, Funahashi_Kijima_2017-1} or \cite[Section 7.7]{Alos_Garcia_Lorite_2021}.

    \item[(ii)] In rough volatility models (such as rough Stein-Stein \cite{Abi_Jaber_2022, Harms_Stefanovits_2019} or rough Heston \cite{EuchGatheralRosenbaum2018, Euch_Rosenbaum_2018, EuchRosenbaum2019}), there is often no transparent procedure of transition between physical and pricing measures: it is often not clear whether the volatility process $\sigma = \{\sigma(t),~t\in[0,T]\}$ hits zero. Hence the integral $\int_0^T \frac{1}{\sigma^2(s)}ds$ that is typically present in densities of martingale measures may be undefined.

    \item[(iii)] In general, stochastic volatility models are susceptible to moment explosions, i.e. the expectations of the prices $\mathbb E[S^2(T)]$ may be infinite for all big enough $T$, see e.g. \cite{AP2006, Gerhold_Gerstenecker_Pinter_2019, Keller-Ressel_2011}  or the article ``\textit{Moment Explosions}'' in \cite{Cont_2010}. Such a technical property can be detrimental in many ways: it may result in infinite prices \cite[Section 8]{AP2006}, invalidate error estimates in numerical schemes \cite[Section 4.2]{Alfonsi_2010}, and cause infinite expected utility in optimization problems \cite{Kallsen_Muhle-Karbe_2010}. In addition, infinite second moments of prices rule out quadratic hedging tools which is a big disadvantage for models that normally produce incomplete markets.
\end{itemize}
Moreover, even if a fractional model somehow overcomes the problems (i)--(iii), it still produces substantial difficulties from the point of view of stochastic methods, especially in optimization problems. As an example, consider a mean-variance hedging problem of the form 
\begin{equation}\label{eq: J}
    \inf_{u} J(u) := \inf_{u}\mathbb E \left[\left(F - \int_0^T u(s) dX(s)\right)^2\right],
\end{equation}
where $X$ is a square-integrable non-Markovian martingale w.r.t. some filtration $\mathbb F = \{\mathcal F_t,~t\in[0,T]\}$ denoting a discounted risk-neutral price, $F = f(X(T))$ is a square integrable financial claim and the infimum in \eqref{eq: J} is taken over all $\mathbb F$-adapted $X$-integrable strategies. From the theoretical perspective, the problem \eqref{eq: J} is well understood: for martingale discounted prices $X$, the \textit{existence} of the solution to \eqref{eq: J} is guaranteed by the celebrated \textit{Galtchouk-Kunita-Watanabe decomposition theorem} (see e.g. \cite{Pham_2000, Schweizer_2001}). Furthermore, \cite{DiNunno} gives the \textit{explicit representation} of the optimal hedging strategy: according to \cite[Theorem 2.1]{DiNunno}, the optimal hedging portfolio $u$ minimizing \eqref{eq: J} can be written as the \textit{non-anticipating derivative} $\mathfrak D F$ defined as the $L^2$-limit of simple processes
\begin{equation}\label{eq: NA-derivative introduction}
\begin{gathered}
    \mathfrak D F := L^2\text{-}\lim_{|\pi| \to 0} u_\pi, \quad
    u_{\pi} := \sum_{k=0}^{n-1} \frac{\mathbb E \left[ (X(t_{k+1}) - X(t_k))F~|~\mathcal F_{t_k}\right]}{\mathbb E\left[ (X(t_{k+1}) - X(t_k))^2~|~\mathcal F_{t_k}\right]} \mathbbm 1_{(t_k, t_{k+1}]},
\end{gathered}    
\end{equation}
where $|\pi| := \max_k(t_k - t_{k-1})$ denotes the mesh of the partition $\pi = \{0 = t_0 < t_1 < ... < t_n = T\}$. Note that \eqref{eq: NA-derivative introduction} is explicit in the sense that the hedge is written only in terms of the discounted price model, the information flow of reference, and the claim $F$. Nevertheless, the practical use of \eqref{eq: NA-derivative introduction} is still limited: analytical expressions for $\mathfrak D F$ are usually impossible to obtain whereas dependence of $X$ on the past significantly complicates numerical computation of the conditional expectations in \eqref{eq: NA-derivative introduction} by e.g. Monte Carlo methods.

\paragraph{Modeling framework: outline and main features.}

In this paper, we consider the numerical aspect of the mean-variance hedging problem \eqref{eq: J} within the \textit{Sandwiched Volterra Volatility} (\textit{SVV}) model initially introduced in \cite{DNMYT2022, DN_YT_power_law_2023} in the option pricing setting. Namely, we assume that the price $S = \{S(t),~t\in[0,T]\}$ of a risky asset has the risk-free dynamics of the form
\begin{align}
    S(t) &= S(0) + \int_0^t \nu(s) S(s)ds + \int_0^t Y(s) S(s) \left(\rho dB_1(s) + \sqrt{1 - \rho^2} dB_2(s)\right), \label{eq: S}
    \\
    Y(t) &= Y(0) + \int_0^t b(s, Y(s))ds + Z(t),\label{eq: Y}
    \\
    X(t) &= \exp\left\{-\int_0^t \nu(s)ds\right\}S(t),\label{eq: X}
\end{align}
where $S(0)$ and $Y(0)$ are given constants, $\nu\in C([0,T])$ is a deterministic instantaneous interest rate, $(B_1, B_2) = \{(B_1(t),B_2(t)),~t\in[0,T]\}$ is a standard 2-dimensional Brownian motion, $\rho\in(-1,1)$, $Z(t) = \int_0^t \mathcal K(t,s) dB_1(s)$ is a H\"older continuous Gaussian Volterra process. The drift $b = b(t,y)$ in \eqref{eq: Y} is unbounded and has an explosive growth to $+\infty$ whenever $y\to\varphi(t)+$ and explosive decay to $-\infty$ whenever $y \to \psi(t)-$, where $0 < \varphi < \psi$ are two deterministic H\"older continuous functions. This structure ensures that the volatility process $Y = \{Y(t),~t\in[0,T]\}$ is \textit{sandwiched} between $\varphi$ and $\psi$, i.e.
\begin{equation}\label{eq: sandwich intro}
    0 < \varphi(t) < Y(t) < \psi(t), \quad t\in [0,T].
\end{equation}
The SVV model \eqref{eq: S}--\eqref{eq: X} has several important advantages.
\begin{itemize}
    \item In line with the recent papers \cite{Jaber_Illand_Shaun_Li_2022, Merino_Pospisil_Sobotka_Sottinen_Vives_2021}, the noise driving the volatility is a general Gaussian Volterra process. This family is very broad and allows bypassing problem (i) by e.g. using several fBms with different Hurst indices as suggested in \cite{Funahashi_Kijima_2017-1} and \cite[Section 7.7]{Alos_Garcia_Lorite_2021} or utilizing multifractional Brownian motion with time-varying local regularity as in \cite{Ayache_Peng_2012, CLV_2014}. For example, according to \cite[Theorem 4.8 and Example 4.9]{DN_YT_power_law_2023}, \eqref{eq: S}--\eqref{eq: X} can reproduce the IV skew power law and long memory simultaneously if $\mathcal K$ is a linear combination of fractional kernels.

    \item The lower bound in \eqref{eq: sandwich intro} guarantees that the volatility $Y$ stays strictly positive. In general, this property allows modeling the market under the physical measure and switching to the pricing measure when required. In this case, \cite[Subsection 2.2]{DNMYT2022} gives the full description of all equivalent local martingale measures.

    \item The upper bound in \eqref{eq: sandwich intro} ensures that the price $S$ has moments of all orders, which allows us to consider the quadratic hedging problem \eqref{eq: J} in the first place.

    \item In addition to addressing problems (i)--(iii), the SVV model has efficient simulation schemes preserving the property \eqref{eq: sandwich intro} \cite{DNMYT2022-1} and generates Malliavin differentiable volatility and price \cite{DNMYT2022, DN_YT_power_law_2023}.
\end{itemize}

We tackle the numerical computation of \eqref{eq: NA-derivative introduction} in three stages. 
\begin{itemize}
    \item[I.] First, we construct a finite-dimensional Markovian approximation to the SVV model \eqref{eq: S}--\eqref{eq: Y} by using the approach in the spirit of \cite{AJE2019, AbiJaber_Miller_Pham_2021, Bauerle_Desmettre_2020, CarmonaCoutin1998}. Namely, we approximate $\mathcal K$ with a sequence of degenerate kernels $\{\mathcal K_m,~m\ge 1\}$
    \[
        \mathcal K_m(t,s) = \sum_{i=0}^{m} e_{m,i}(t)f_{m,i}(s)\mathbbm 1_{s<t}
    \]
    and then exploit the Markovianity of the $(m+2)$-dimensional process $(X_m, Y_m, U_{m,0}, ..., U_{m,m})$, where
    \begin{equation}\label{eq: model approx intro}
    \begin{aligned}
        X_m(t) &= X(0) + \int_0^t Y_m(s) X_m(s) \left(\rho dB_1(s) + \sqrt{1 - \rho^2} dB_2(s)\right),
        \\
        Y_m(t) &= Y(0) + \int_0^t b(s, Y_m(s))ds + \sum_{i=1}^{m} e_{m,i}(t) U_{m,i}(t),
        \\
        U_{m,i}(t) &= \int_0^t f_{m,i}(s) dB_1(s),\quad i=0,1,...,m.
    \end{aligned}
    \end{equation}

    \item[II.] Second, we plug \eqref{eq: model approx intro} in the conditional expectations from \eqref{eq: NA-derivative introduction} and prove that the result converges to $\mathfrak D F$ in $L^2$.

    \item[III.] Finally, we employ the Markovian structure of \eqref{eq: model approx intro} to numerically estimate conditional expectations $\mathbb E \left[ (X_m(t_{k+1}) - X_m(t_k))F(X_m(T))~|~\mathcal F_{t_k}\right]$ and $\mathbb E\left[ (X_m(t_{k+1}) - X_m(t_k))^2~|~\mathcal F_{t_k}\right]$. We propose two algorithms: Nested Monte Carlo (NMC) and Least Squares Monte Carlo (LSMC).
\end{itemize}
Note that the convergence of \eqref{eq: model approx intro} as $m\to\infty$ is not trivial due to the explosive nature of the drift $b$. For this reason, we have to use our own technique based on pathwise estimates from \cite{DNMYT2020}.

\paragraph{Structure of this work.} 
Section \ref{s: Preliminaries and Assumptions} contains the detailed description of the SVV model \eqref{eq: S}--\eqref{eq: X} and gathers some necessary known results. Section \ref{s: Markovian approximation} is devoted to Stage I of our plan above: we introduce the approximation \eqref{eq: model approx intro} and prove its convergence to the original SVV model. In Section \ref{s: NA-derivative estimates}, Stage II is realized: we construct an approximation of the optimal hedging strategy $\mathfrak D F$ and study its convergence. Stage III is implemented in Section \ref{s: Monte Carlo approximation of the hedging strategy} which concentrates on numerical algorithms for computing the optimal hedge and describes two Monte Carlo approaches for the computation of the non-anticipating derivative; the results are illustrated by simulations. Appendices \ref{appendix: proof of approximation of volatility} and \ref{appendix: estimates for price increments} contain the proofs of some technical results.

\section{Sandwiched Volterra Volatility Model}\label{s: Preliminaries and Assumptions}

Hereafter we give detailed specifications for the model \eqref{eq: S}--\eqref{eq: X} and collect some fundamental results for the upcoming discussion. Before proceeding to the main content, let us introduce some notation used in the sequel.

\begin{notation*}\hfill
    \begin{enumerate}
        \item The model \eqref{eq: S}--\eqref{eq: Y} is defined on a complete probability space $(\Omega, \mathcal F, \mathbb P)$ with a $\mathbb P$-augmented filtration $\mathbb F = \{\mathcal F_t,~t\in[0,T]\}$, $T<\infty$,  generated by the 2-dimensional Brownian motion $(B_1, B_2) = \{(B_1(t), B_2(t)),~t\in [0,T]\}$.
    
        \item  In what follows, by $C$ we will denote any positive deterministic constant the exact value of which is not important in the context. Note that $C$ may change from line to line (or even within one line).
        
        \item Throughout the paper, $L^2(\mathbb P)$ will denote a standard space of square-integrable random variables. For a given square-integrable martingale $\eta = \{\eta(t),~t\in[0,T]\}$, the Hilbert space of stochastic processes such that
        \[
            \lVert u \rVert_{L^2(\mathbb P \times [\eta])} := \left(\mathbb E \left[ \int_0^T u^2(s) d[\eta](s) \right]\right)^{\frac{1}{2}} < \infty 
        \]
        will be denoted by $L^2(\mathbb P \times [\eta])$. 
    \end{enumerate}
\end{notation*}

\subsection{Volterra noise}\label{subsec: noise}
Consider a function $\mathcal K$: $[0,T]^2 \to \mathbb R$ that satisfies the following assumption.

\begin{assum}{(K)}\label{assum: kernel} The function $\mathcal K$: $[0,T]^2 \to \mathbb R$ is a \textit{Volterra kernel}, i.e. $\mathcal K(t,s) = 0$ whenever $t<s$, and the following conditions hold:
    \begin{itemize}
    \item[\textbf{(K1)}] $\mathcal K$ is square integrable, i.e.
        \[
            \int_0^T \int_0^T \mathcal K^2(t,s)ds dt = \int_0^T \int_0^t \mathcal K^2(t,s)dsdt < \infty;
        \]
    \item[\textbf{(K2)}] there exists a constant $H \in (0,1)$ such that, for any $\lambda \in (0,H)$,
        \[
            \int_0^t (\mathcal K(t,u) - \mathcal K(s, u))^2 du \le \ell_\lambda^2 |t-s|^{2\lambda}, \quad 0\le s\le t\le T,
        \]
        where $\ell_\lambda > 0$ is a constant, possibly dependent on $\lambda$.
\end{itemize}
\end{assum}
\noindent For a fixed $\mathcal K$ satisfying Assumption \ref{assum: kernel}, we define a Gaussian Volterra process $Z=\{Z(t),~t\in[0,T]\}$ as
\begin{equation}\label{eq: Z}
    Z(t) := \int_0^t \mathcal K(t,s)dB_1(s), \quad t\in [0,T],
\end{equation}
and immediately remark that
\begin{itemize}
    \item[1)] \textbf{(K1)} ensures that the stochastic integral in \eqref{eq: Z} is well-defined for all $t\in[0,T]$,
        
    \item[2)] \textbf{(K2)} guarantees (see e.g. \cite[Theorem 1 and Corollary 4]{ASVY2014}) that $Z$ has a modification that is H\"older continuous of any order $\lambda\in(0,H)$. Moreover, for any $\lambda\in(0,H)$, the positive random variable $\Lambda = \Lambda_\lambda$ such that
        \begin{equation}\label{eq: Lambda}
            |Z(t) - Z(s)| \le \Lambda |t-s|^\lambda, \quad t,s\in[0,T],
        \end{equation}
        can be chosen to have moments of all orders. In what follows, such H\"older continuous modification as well as $\Lambda$ with all the moments will be used.
\end{itemize}

\noindent Before proceeding further, let us give two examples of kernels satisfying Assumption \ref{assum: kernel}.

\begin{example}\label{ex: Holder kernel}{\textit{(H\"older continuous kernels)}}
    Let $\mathcal K(t,s) = \mathcal K(t-s)\mathbbm 1_{0\le s \le t \le T}$ for some $H$-H\"older continuous function $\mathcal K$, $H \in (0,1)$. Then
    \begin{align*}
        \int_0^t &(\mathcal K(t,u) - \mathcal K(s, u))^2 du = \int_0^s (\mathcal K(t - u) - \mathcal K(s - u))^2 du + \int_s^t (\mathcal K(t-u) - \mathcal K(0) +\mathcal K(0))^2 du
        \\
        & \le C \int_0^s |t-s|^{2H} du + C \int_0^{t-s} u^{2H}  du + C \int_0^{t-s} \mathcal K^2(0) du \le C|t-s|^{2\left( H \wedge \frac 1 2 \right)},
    \end{align*}
    i.e. \textbf{(K2)} is satisfied and the process $Z(t) = \int_0^t \mathcal K(t-s) dB_1(s)$, $t\in [0,T]$, has a modification that is H\"older continuous up to the order $H \wedge \frac 1 2$. Furthermore, if additionally $\mathcal K(0) = 0$, then it is easy to see that
    \[
         \mathbb E [(Z(t) - Z(s))^2] = \int_0^t (\mathcal K(t,u) - \mathcal K(s, u))^2 du \le C|t-s|^{2H},
    \]
    i.e. $Z$ has a modification that is H\"older continuous up to the order $H$.
\end{example}

\begin{example}\label{ex: fractional kernel}\label{ex: RL-fBm}{\textit{(Fractional kernel)}}
    Let $\mathcal K(t,s) = \frac{1}{\Gamma\left(H+\frac{1}{2}\right)} (t-s)^{H - \frac{1}{2}}\mathbbm 1_{0\le s \le t \le T}$ with $H\in\left(0, 1 \right)$. The process 
    \[
        Z(t) = \frac{1}{\Gamma\left(H+\frac{1}{2}\right)}\int_0^t (t-s)^{H-\frac{1}{2}} dB_1(s), \quad t\in [0,T],
    \]
    is called the \textit{Riemann-Liouville fractional Brownian motion} (RL-fBm) and it is well-known (see e.g. \cite[Lemma 3.1]{WX2022}) that, for any $s<t$,
    \begin{align*}
        \mathbb E [(Z(t) - Z(s))^2] & = \frac{1}{\Gamma^2\left(H+\frac{1}{2}\right)}\left(\int_0^s \left((t - u)^{H - \frac{1}{2}} - (s - u)^{H - \frac{1}{2}}\right)^2 du + \int_0^{t-s} u^{2H - 1} du\right)
        \\
        & \le C|t-s|^{2H}
    \end{align*}
    for some constant $C$. Thus \textbf{(K2)} holds and the RL-fBm has a modification that is H\"older continuous up to the order $H$.
\end{example}

\subsection{Sandwiched volatility} 

For a given pair $\varphi$, $\psi$: $[0,T] \to \mathbb R$ such that $\varphi(t) < \psi(t)$, $t\in[0,T]$, and any $a_1$, $a_2 \ge 0$, denote
\[
    \mathcal D_{a_1,a_2} := \{(t,y)\in[0,T]\times\mathbb R_+,~y\in(\varphi(t) + a_1, \psi(t) - a_2)\}.
\]
Consider a stochastic process $Y = \{Y(t),~t\in[0,T]\}$ defined by \eqref{eq: Y}, where the initial value $Y(0)$ and the drift $b$ satisfy the following assumption.

\begin{assum}{(Y)}\label{assum: drift} There exist $H$-H\"older continuous functions $\varphi$, $\psi$: $[0,T]\to\mathbb R$, $0 < \varphi(t) < \psi(t)$, $t\in[0,T]$, with $H$ being the same as in \textbf{(K2)} such that
    \begin{itemize}
        \item[\textbf{(Y1)}] $Y(0)$ is deterministic and $0 < \varphi(0) < Y(0) < \psi(0)$,
        \item[\textbf{(Y2)}] $b$: $\mathcal D_{0,0} \to \mathbb R$ is continuous and, moreover, for any $\varepsilon \in \left(0, \min\left\{1, \frac{1}{2}\lVert \psi - \varphi \rVert_\infty\right\}\right)$
        \[
            |b(t_1,y_1) - b(t_2, y_2)| \le \frac{c_1}{\varepsilon^{p}} \left(|y_1 - y_2| + |t_1 - t_2|^H \right), \quad (t_1, y_1), (t_2,y_2) \in \mathcal D_{\varepsilon, \varepsilon},
        \]
        where $c_1 > 0$ and $p>1$ are given constants,
        \item[\textbf{(Y3)}] 
        \begin{equation}\label{eq: explosion in b}
        \begin{aligned}
            b(t, y) &\ge \frac{c_2}{(y - \varphi(t))^{\gamma}}, &\quad (t,y) \in \mathcal D_{0,0}\setminus \mathcal D_{y_*, 0},
             \\
             b(t, y) &\le -\frac{c_2}{(\psi(t) - y)^{\gamma}}, &\quad (t,y) \in \mathcal D_{0,0}\setminus \mathcal D_{0, y_*},
        \end{aligned}
        \end{equation}
        where $y_*$, $c_2 > 0$ are given constants and $\gamma > \frac{1}{H} - 1$,
        \item[\textbf{(Y4)}] the partial derivative $\frac{\partial b}{\partial y}$ with respect to the spacial variable exists, is continuous and
        \[
            -c_3 \left( 1 + \frac{1}{(y - \varphi(t))^q} + \frac{1}{(\psi(t) - y)^q} \right) < \frac{\partial b}{\partial y}(t, y) < c_3, \quad (t,y) \in \mathcal D_{0,0},
        \]
        for some $c_3 > 0$ and $q>0$.
    \end{itemize}
\end{assum}

\noindent Processes with such drifts were studied in detail in \cite{DNMYT2020}. For reader's convenience, we collect some relevant results form \cite{DNMYT2020} regarding the volatility \eqref{eq: Y} in the statement below.

\begin{theorem}\label{th: properties of Y}
    Let Assumptions \ref{assum: kernel} and \ref{assum: drift} hold. Then
    \begin{itemize}
        \item[1)] the SDE \eqref{eq: Y} has a unique strong solution $Y$ and \eqref{eq: double sandwich 1} holds with probability 1;
        
        \item[2)] for any $\lambda\in\left(\frac{1}{1+\gamma},H\right)$, where $H$ is from \textbf{(K2)}, there exist deterministic constants $L_1$ and $L_2 > 0$ (depending only on $Y(0)$, the shape of $b$ and $\lambda$) such that, with probability 1,
        \begin{equation}\label{eq: upper and lower bounds for sandwiched volatility, general case}
            \varphi(t) + \frac{L_1}{ ( L_2 + \Lambda )^{\frac{1}{\gamma \lambda + \lambda - 1}} } \le Y(t) \le \psi(t) - \frac{L_1}{ ( L_2 + \Lambda )^{\frac{1}{\gamma \lambda + \lambda - 1}} }, \quad  \quad t\in[0,T],
        \end{equation}
        where $\Lambda$ is from \eqref{eq: Lambda} and $\gamma$ is from \textbf{(Y3)}; in particular,
        \begin{equation}\label{eq: double sandwich 1}
            \varphi(t) < Y(t) < \psi(t), \quad t \in [0,T];
        \end{equation}
        
        \item[3)] for any $r \in \mathbb R$,
        \[
            \mathbb E\left[ \sup_{t\in[0,T]} \frac{1}{(Y(t) - \varphi(t))^r} \right] < \infty, \quad \mathbb E\left[ \sup_{t\in[0,T]} \frac{1}{(\psi(t) - Y(t))^r} \right] < \infty.
        \]
    \end{itemize}
\end{theorem}
\begin{remark}
    Having in mind property \eqref{eq: double sandwich 1}, we will refer to the process $Y = \{Y(t),~t\in[0,T]\}$ defined by \eqref{eq: Y} as a \textit{sandwiched} process.
\end{remark}

\noindent In addition to Theorem \ref{th: properties of Y}, we will also use the following result which can be found in \cite[Lemma 3.6]{DNMYT2022-1}.
\begin{lemma}
    Let Assumptions \ref{assum: kernel} and \ref{assum: drift} hold. Then, for each $\lambda\in \left(0, H\right)$ there exists a random variable $\Upsilon = \Upsilon_\lambda>0$ such that $\mathbb E[\Upsilon^r] < \infty$ for any $r>0$ and
    \[
        |Y(t) - Y(s)| \le \Upsilon |t-s|^\lambda, \quad t,s\in[0,T].
    \]
\end{lemma}

\subsection{Price process} 

Let Assumptions \ref{assum: kernel} and \ref{assum: drift} hold, $Z$ be a Gaussian Volterra process \eqref{eq: Z}, and $Y = \{Y(t),~t\in[0,T]\}$ be the sandwiched process \eqref{eq: Y} such that \eqref{eq: double sandwich 1} holds for some functions $0< \varphi < \psi$. We now consider the model introduced in Section \ref{s: introduction paper 5} with a risk-free asset $e^{\int_0^t \nu(s) ds}$ and a risky asset $S = \{S(t),~t\in[0,T]\}$ defined by \eqref{eq: S}, where $S(0) > 0$ is some deterministic constant and $\rho \in (-1,1)$. Clearly, the discounted price process $X$ defined by \eqref{eq: X} satisfies the SDE
\begin{equation}\label{eq: X dynamics}
    X(t) = X(0) + \int_0^t Y(s) X(s) \left( \rho dB_1(t) + \sqrt{1 - \rho^2} dB_2(t) \right), \quad t\in [0,T],
\end{equation}
where, for simplicity, we denote $X(0) := S(0)$. Moreover, the following Theorem holds.

\begin{theorem}\label{th: properties of price}{({\cite[Theorem 2.6]{DNMYT2022})}}
    \begin{itemize}
        \item[1)] Equations \eqref{eq: S} and \eqref{eq: X dynamics} both have unique strong solutions and
        \begin{align*}
             S(t) &= S(0) \exp\left\{ \int_0^t \left(\nu(s) - \frac{Y^2(s)}{2} \right)ds + \int_0^t Y(s) \left( \rho dB_1(t) + \sqrt{1 - \rho^2} dB_2(t) \right) \right\};
             \\
             X(t) &= X(0) \exp\left\{ - \int_0^t \frac{Y^2(s)}{2} ds + \int_0^t Y(s) \left( \rho dB_1(t) + \sqrt{1 - \rho^2} dB_2(t) \right) \right\}.
        \end{align*}
        
        \item[2)] For any $r\in\mathbb R$,
        \begin{equation}\label{eq: sup of S and X}
            \mathbb E \left[ \sup_{t\in[0,T]} S^r(t) \right] < \infty, \quad \mathbb E \left[ \sup_{t\in[0,T]} X^r(t) \right] < \infty.
        \end{equation}
    \end{itemize}
\end{theorem}

\begin{remark}
    Theorem \ref{th: properties of price} implies that $X = \{X(t),~t\in[0,T]\}$ is a square integrable martingale w.r.t. the filtration $\mathbb F = \{\mathcal F_t,~t\in[0,T]\}$ generated jointly by $B_1$ and $B_2$.
\end{remark}

\begin{remark}
    We stress that \eqref{eq: sup of S and X} does not hold in general even in classical stochastic volatility models such as the Heston model (see e.g. \cite{AP2006}). This has repercussions in the possibilities of numerical simulations as well as in various applications such as pricing hedging or portfolio optimization. In our case, the existence of moments of all orders is achieved thanks to the boundedness of $Y$. 
\end{remark}

\noindent We conclude this Section with a result on regularity of the probability laws of $S$ and $X$ presented in \cite{DNMYT2022}.

\begin{theorem}{{\cite[Corollary 3.8]{DNMYT2022}}}\label{th: absolute continuity of laws of S and X} For any $t\in(0,T]$, the law of $S(t)$ (and consequently of $X(t) = e^{-\int_0^t \nu(s) ds} S(t)$) has continuous and bounded density.
\end{theorem}

\section{Markovian approximation of the SVV model}\label{s: Markovian approximation}

In general, the SVV model \eqref{eq: S}--\eqref{eq: Y} is non-Markovian because of the Volterra Gaussian noise in the volatility dynamics. As discussed in the Introduction, while it can be viewed as a positive feature from the modeling perspective, non-Markovianity entails numerous numerical problems. In principle, one could construct Markovian approximations for the SVV model based on the following idea taking inspiration from \cite{AJE2019, CarmonaCoutin1998}. It is evident that the original kernel $\mathcal K$ in \eqref{eq: Y} can be approximated in $L^2$ by a sequence of \textit{degenerate} Volterra kernels
\[
    \mathcal K_m(t,s) = \sum_{i=0}^{m} e_{m,i}(t)f_{m,i}(s)\mathbbm 1_{s<t},
\]
where $\int_0^T e^2_{m,i}(t)dt < \infty$, $\int_0^T f^2_{m,i}(t)dt < \infty$, $i=0,...,m$, $m\ge 0$. Then one can attempt to define the approximated SVV model as
\begin{align}
    S_m(t) &= S(0) + \int_0^t \nu(s) S_m(s) ds + \int_0^t Y_m(s) S_m(s) dW(s),\label{eq: price dynamics approximated}
    \\
    Y_m(t) &= Y(0) + \int_0^t b(s, Y_m(s))ds + Z_m(t),\label{eq: sandwiched approximation}
    \\
    X_m(t) &= e^{-\int_0^t \nu(s) ds }S_m(t),\label{eq: X approximated}
\end{align}
where $Z_m$ is given by
\begin{equation}\label{eq: approximated noise}
    Z_m(t) := \int_0^t \mathcal K_m(t,s) dB_1(s) = \sum_{i=0}^{m} e_{m,i}(t)\int_0^t f_{m,i}(s)dB_1(s), \quad t\in[0,T].
\end{equation}
In this case, the $(m+2)$-dimensional process $(S_m, Y_m, U_{m,0}, ..., U_{m,m})$, where 
\begin{equation}\label{eq: finite-dimensional Markovian approximation}
\begin{aligned}
    U_{m,i}(t) &= \int_0^t f_{m,i}(s) dB_1(s), \quad i = 0,1,...,m,
\end{aligned}
\end{equation}
is Markovian w.r.t. the filtration $\mathbb F = \{\mathcal F_t,~t\in[0,T]\}$ generated by $(B_1, B_2)$ and, \textit{intuitively}, should approximate the original non-Markovian SVV model if $\mathcal K_m$ is close enough to $\mathcal K$. However, such intuition is not simple to back up analytically due to the explosion \eqref{eq: explosion in b} of the drift $b$:
\begin{itemize}
    \item Assumption \textbf{(Y3)} demands $\gamma > \frac{1}{H}-1$, i.e. the order of the drift explosion in \eqref{eq: explosion in b} is coupled with the H\"older regularity of the Volterra noise $Z$ -- changing the kernel may violate \textbf{(Y3)} and hence \eqref{eq: sandwiched approximation} may not have a solution;

    \item it is not clear why $Y_m$ approximates $Y$ in $L^2$: the standard reasoning via Gronwall's inequality fails, again because of the explosive drift.
\end{itemize}

\noindent The goal of this section is to overcome the problems mentioned above and study the convergence of approximations \eqref{eq: price dynamics approximated}--\eqref{eq: X approximated}.

\subsection{Admissible approximating kernels}

As noted above, we want to replace the kernel $\mathcal K$ with its appropriate $L^2$-approximations $\mathcal K_m$, $m\ge 1$, and then consider stochastic processes $Z_m = \{Z_m(t),~t\in[0,T]\}$ and $Y_m = \{Y_m(t),~t\in[0,T]\}$ defined by \eqref{eq: approximated noise} and \eqref{eq: sandwiched approximation} respectively. However, we need to make sure that the solution to \eqref{eq: sandwiched approximation} exists, i.e. Theorem \ref{th: properties of Y} holds, for all $m\ge 1$. Moreover, we will need that bounds of the type \eqref{eq: upper and lower bounds for sandwiched volatility, general case} hold for each $Y_m$ with constants $L_1$, $L_2$ that do not depend on $m$. In order to ensure that, we introduce the following assumption.

\begin{assum}{{(Km)}}\label{assum: approx kernels} 
    The sequence of kernels $\{\mathcal K_m, m\ge 1\}$ is such that each function $\mathcal K_m$: $[0,T]^2 \to \mathbb R$ is a Volterra kernel, i.e. $\mathcal K_m(t,s) = 0$ whenever $t<s$, and
    \begin{itemize}
        \item[\textbf{(Km1)}] for each $m\ge 1$, $\mathcal K_m$ is square integrable, i.e.
        \[
            \int_0^T \int_0^T \mathcal K_m^2(t,s)ds dt = \int_0^T \int_0^t \mathcal K_m^2(t,s)dsdt < \infty;
        \]
        
        \item[\textbf{(Km2)}] for each $m\ge 1$ and $\lambda \in (0,H)$ with $H$ being from \textbf{(K2)},
        \[
            \int_0^t \left( \mathcal K_m(t,u) - \mathcal K_m(s,u) \right)^2 du \le \ell^2_{\lambda} (t-s)^{2\lambda}, \quad 0\le s \le t\le T,
        \]
        where $\ell_{\lambda} > 0$ is a constant that possibly depends on $\lambda$ but does not depend on $m$.
    \end{itemize}
\end{assum}

\noindent In general, it is not clear whether a given $L^2$-approximation $\{\mathcal K_m,~m\ge 1\}$ of the kernel $\mathcal K$ satisfies Assumption \ref{assum: approx kernels}, so one has to check it separately. Luckily, for kernels from Examples \ref{ex: Holder kernel} and \ref{ex: fractional kernel} above, such approximations exist and have relatively simple shapes. 

\begin{example}[\textit{H\"older continuous kernels}]\label{ex: Holder kernel approx}
    Assume that $\mathcal K(t,s) = \mathcal K(t-s) \mathbbm 1_{s<t}$ with $\mathcal K(0) = 0$ and $\mathcal K \in C^H([0,T])$ for some $H\in (0,1)$, i.e.
    \begin{equation}\label{eq: Holder condition for kernel}
        |\mathcal K(t) - \mathcal K(s)| \le C_H |t-s|^H, \quad s,t \in [0,T].
    \end{equation}
    Let $\mathcal K_m$ be the corresponding Bernstein polynomial of order $m$ defined as
    \begin{equation}\label{eq: Bernstein polynomial}
    \begin{aligned}
        \mathcal K_m(t) &= \frac{1}{T^m} \sum_{i=0}^m \mathcal K\left(\frac{Ti}{m}\right) \binom{m}{i} t^i (T-t)^{m-i} = \sum_{i=0}^m \left( \frac{1}{T^i} \sum_{j=0}^i (-1)^{i-j} \mathcal K\left(\frac{Tj}{m}\right) \binom{m}{j} \binom{m-j}{i-j} \right) t^i
        \\
        & =: \sum_{i=0}^m \kappa_{m,i} t^i
    \end{aligned}    
    \end{equation}
    and
    \begin{equation}\label{eq: K Bernstein}
    \begin{aligned}
        \mathcal K_m(t,s) &=\mathcal K_m(t-s)\mathbbm 1_{s<t} = \sum_{i=0}^m \kappa_{m,i} \left(\sum_{j=0}^i  \binom{i}{j} (-1)^{i-j} t^j s^{i-j}\right) \mathbbm 1_{s<t} 
        \\
        &= \sum_{j=0}^m t^j \left(\sum_{i=j}^m \kappa_{m,i} \binom{i}{j} (-1)^{i-j}  s^{i-j}\right) \mathbbm 1_{s<t} =: \sum_{i=0}^{m} e_{m,i}(t)f_{m,i}(s)\mathbbm 1_{s<t}.
    \end{aligned}   
    \end{equation}
    By \cite[Proposition 2]{Mathe1999}, 
    \[
        |\mathcal K_m(t) - \mathcal K_m(s)| \le C_H |t-s|^H, \quad s,t \in [0,T],
    \]
    where $C_H$ is the same as in \eqref{eq: Holder condition for kernel} for all $m\ge 1$. Since $\mathcal K(0) = 0$, $\mathcal K_m(0) = \kappa_{m,0} = 0$ and we can write for any $0\le s \le t \le T$
    \begin{align*}
        \int_0^t (\mathcal K_m (t-u) &- \mathcal K_m(s-u))^2 du + \int_s^t \mathcal K_m^2(t-u)du \le C_H^2 T |t-s|^{2H} + \int_0^{t-s} \mathcal K_m^2(v) dv
        \\
        & \le C_H^2 T |t-s|^{2H} + C_H^2 \int_0^{t-s} v^{2H} dv \le  C_H^2 T |t-s|^{2H} + \frac{C_H^2}{2H+1} |t-s|^{2H+1}
        \\
        & \le C|t-s|^{2H},
    \end{align*}
    where $C$ does not depend on $m$, i.e. Assumption \ref{assum: approx kernels} holds. Moreover, \cite[Theorem 1]{Mathe1999} gives the estimate
    \[
        \sup_{t\in[0,T]} | \mathcal K(t) - \mathcal K_m(t) | \le C m^{-\frac{H}{2}},
    \]
    which implies that
    \begin{equation}\label{eq: Bernstein kernel error}
        \lVert \mathcal K - \mathcal K_m \rVert_{L^2([0,T])} \le C m^{-\frac{H}{2}} \to 0, \quad m \to \infty.
    \end{equation}
\end{example}

\begin{remark}
    Note that each process
    \[
        Z_m(t) = \int_0^t \mathcal K_m(t-s)dB_1(s) = \sum_{i=0}^m \kappa_{m,i} \int_0^t (t-s)^i dB_1(s), \quad t\in [0,T], 
    \]
    is actually H\"older continuous of \textit{any} order $\lambda\in(0,1)$, i.e. the solution to the SDE
    \[
        Y_m(t) = Y(0) + \int_0^t b(s, Y_m(s))ds + Z_m(t), \quad t\in [0,T],
    \]
    exists and is unique for any value $\gamma > 0$ in \textbf{(Y3)} provided that $\varphi$ and $\psi$ are both Lipschitz.
\end{remark}

\begin{example}[\textit{Fractional kernel}]\label{ex: fractional kernel approx}
    Let $H\in \left(0, \frac{1}{2}\right)$ and 
    \begin{equation}\label{eq: rough fractional kernel}
        \mathcal K(t,s) = \mathcal K(t-s) \mathbbm 1_{s<t} = \frac{(t-s)^{H-\frac{1}{2}}}{\Gamma\left(H + \frac{1}{2}\right)} \mathbbm 1_{s<t},
    \end{equation}
    i.e. $Z(t) = \int_0^t \mathcal K(t,s) dB_1(s)$ is a RL-fBm as discussed in Example \ref{ex: RL-fBm}.
    For such a kernel, an appropriate sequence $\{\mathcal K_m,~m\ge 1\}$ of degenerate $L^2$-approximations satisfying Assumption \ref{assum: approx kernels} is presented in \cite[Section 3]{AJE2019}. Namely, let $\mu$ be a measure on $\mathbb R_+$ defined as
    \[
        \mu(d\alpha) := \frac{\alpha^{-H-\frac{1}{2}}}{\Gamma\left(H+ \frac{1}{2}\right)\Gamma\left(\frac{1}{2} - H\right)}d\alpha,
    \]
    \[
        \sigma_{m,i} := \int_{\tau_{m,i-1}}^{\tau_{m,i}} \mu(d\alpha), \quad \alpha_{m,i} := \frac{1}{\sigma_{m,i}} \int_{\tau_{m, i-1}}^{\tau_{m,i}} \alpha \mu(d\alpha), \quad i = 1,...,m,
    \]
    where $0 = \tau_{m,0} < \tau_{m,1} < ... < \tau_{m,m}$ is such that
    \begin{equation}\label{eq: rough fractional kernel condition on partition}
        \tau_{m,m} \to \infty, \quad \sum_{i=1}^m \int_{\tau_{m, i-1}}^{\tau_{m,i}} (\alpha_{m,i} - \alpha)^2 \mu(d\alpha) \to 0, \quad m\to\infty.
    \end{equation}
    Denote for $m\ge 1$
    \begin{equation}\label{eq: rough fractional kernel approximation}
        \mathcal K_m(t-s)\mathbbm 1_{s<t} = \sum_{i=1}^m \sigma_{m,i}  e^{-\alpha_{m,i}(t-s)}\mathbbm 1_{s<t} = \sum_{i=1}^m \left(\sigma_{m,i} e^{-\alpha_{m,i}t}\right)   e^{\alpha_{m,i}s}\mathbbm 1_{s<t} =: \sum_{i=1}^{m} e_{m,i}(t)f_{m,i}(s)\mathbbm 1_{s<t}.
    \end{equation}
    By \cite[Lemma 5.2]{AJE2019}, the sequence $\{\mathcal K_m,~m\ge 1\}$ satisfies Assumption \ref{assum: approx kernels} and, moreover, by \cite[Proposition 3.3]{AJE2019},
    \[
        \lVert \mathcal K - \mathcal K_m \rVert_{L^2([0,T])} \to 0, \quad m\to\infty.
    \]
    Note that the values $0 = \tau_{m,0} < \tau_{m,1} < ... < \tau_{m,m}$ satisfying \eqref{eq: rough fractional kernel condition on partition} can be chosen in multiple ways. For example, as suggested in \cite[Subsection 3.2]{AJE2019}, one can put
    \begin{equation}\label{eq: possible choice of tau}
        \tau_{m,i} := \frac{1}{T}\left(\frac{\sqrt{10}(1-2H)}{5 - 2H}\right)^{\frac{2}{5}} im^{-\frac{1}{5}}, \quad i = 0, 1, 2, ..., m.    
    \end{equation}
    In this case,
    \begin{equation}\label{eq: fractional kernel error}
        \lVert \mathcal K - \mathcal K_m \rVert_{L^2([0,T])} \le C_H m^{-\frac{4H}{5}},
    \end{equation}
    where  $C_H$ is a positive constant that depends only on the Hurst parameter $H \in \left(0, \frac{1}{2}\right)$.
\end{example}

\begin{remark}\label{rem: OU}
    Putting $e_{m,i}(t) := \sigma_{m,i} e^{-\alpha_{m,i} t}$, $f_{m,i}(s) := e^{\alpha_{m,i} s}$ in Example \ref{ex: fractional kernel approx}, we obtain that each process of the form 
    \[
         e_{m,i}(t) U_{m, i}(t) = \sigma_{m,i}\int_0^t  e^{-\alpha_{m,i}(t-s)} dB_1(s) =: V_{m,i}(t), \quad t\in [0,T], 
    \]
    from \eqref{eq: finite-dimensional Markovian approximation} is an Ornstein-Uhlenbeck process satisfying the SDE
    \[
        V_{m,i}(t) = -\alpha_{m,i} \int_0^t V_{m,i}(s) ds + \sigma_{m,i} B_1(t), \quad t\in [0,T].
    \]
\end{remark}

\noindent Just like in Subsection \ref{subsec: noise}, we see here that Assumption \ref{assum: approx kernels} together with \cite[Theorem 1 and Corollary 4]{ASVY2014} guarantee that each $Z_m$ has a continuous modification that satisfies the $\lambda$-H\"older property for all $\lambda \in (0,H)$. In fact, \textbf{(Km2)} allows to deduce a stronger statement: the H\"older seminorms of $Z_m$, $m\ge 1$, are uniformly bounded in $L^r(\mathbb P)$. The corresponding result is presented in the next lemma.

\begin{lemma}\label{lemma: Holder continuity of approximated noises}
    Let $\mathcal K$ satisfy Assumption \ref{assum: kernel} and the sequence $\{\mathcal K_m, m\ge 1\}$ satisfy Assumption \ref{assum: approx kernels}. Then for any constant $\lambda \in \left( 0, H\right)$ there exist positive random variables $\Lambda_m$, $m\ge 1$, such that for each $m\ge 1$, $r>0$ and $s,t\in[0,T]$
    \[
        |Z_m(t) - Z_m(s)| \le \Lambda_m |t-s|^\lambda,
    \]
    and
    \begin{equation}\label{eq: uniform bound on Holder seminorms of the noise}
        \sup_{m\ge 1} \mathbb E \Lambda_m^r < \infty.
    \end{equation}
\end{lemma}

\begin{proof}
    Fix an arbitrary $\lambda \in \left(0, H\right)$, choose $\alpha > 1$ such that $\lambda + \frac{2}{\alpha} < H$ and denote 
    \begin{align*}
        \Lambda_m &:= A_{\lambda, \alpha} \left( \int_0^T \int_0^T \frac{|Z_m(t) - Z_m(s)|^\alpha}{|t-s|^{\lambda \alpha + 2}} ds dt\right)^{\frac{1}{\alpha}}, \quad m\ge 1,
    \end{align*}
    with $A_{\lambda, \alpha} := 2^{3 + \frac{2}{\alpha}} \left( \frac{\lambda \alpha + 2}{\lambda \alpha} \right)$. By the Garsia-Rodemich-Rumsey inequality (see \cite{GRR} and \cite[Lemma 1.1]{DNMYT2020}), for any $m \ge 1$ and $s,t\in[0,T]$,
    \[
        |Z_m(t) - Z_m(s)| \le \Lambda_m |t-s|^\lambda,
    \]
    and, moreover, the proof of \cite[Theorem 1]{ASVY2014} implies that $\mathbb E \Lambda^r_m < \infty$ for each $r>0$. It remains to show that $\sup_{m\ge 1}\mathbb E \Lambda^r_m < \infty$. Take any $r > \alpha$, and observe that Minkowski integral inequality, Gaussian distribution of $Z_m(t) - Z_m(s)$ and \textbf{(Km2)} yield
    \begin{align*}
        (\mathbb E[\Lambda_m^{r}])^{\frac{\alpha}{r}} & \le A^\alpha_{\lambda, \alpha} \int_0^T \int_0^T \frac{\mathbb (E[|Z_m(t)-Z_m(s)|^r])^{\frac{\alpha}{r}}}{|t-s|^{\lambda \alpha + 2}}  ds dt
        \\
        & = 2^{\frac{\alpha}{2}}\left( \frac{\Gamma\left(\frac{r+1}{2}\right)}{\sqrt{\pi}}\right)^{\frac{\alpha}{r}}  A^\alpha_{\lambda, \alpha} \int_0^T \int_0^T \frac{ \left( \int_0^{t\vee s} \left( \mathcal K_m(t,u) - \mathcal K_m(s,u) \right)^2 du\right)^{\frac{\alpha}{2}}}{|t-s|^{\lambda \alpha + 2}}  ds dt
        \\
        & \le 2^{\frac{\alpha}{2}}\left( \frac{\Gamma\left(\frac{r+1}{2}\right)}{\sqrt{\pi}}\right)^{\frac{\alpha}{r}}  A^p_{\lambda, \alpha} \ell^\alpha_{\lambda + \frac{2}{\alpha}} \int_0^T \int_0^T \frac{ |t-s|^{\lambda \alpha + 2}}{|t-s|^{\lambda \alpha + 2}}  ds dt
        \\
        & \le 2^{\frac{\alpha}{2}} T^2 \left( \frac{\Gamma\left(\frac{r+1}{2}\right)}{\sqrt{\pi}}\right)^{\frac{\alpha}{r}}  A^\alpha_{\lambda, \alpha} \ell^\alpha_{\lambda + \frac{2}{\alpha}}
    \end{align*}
    and thus
    \[
        \sup_{m \ge 1} E[\Lambda_m^{r}] \le 2^{\frac{r}{2}} T^\frac{2r}{\alpha} \frac{\Gamma\left(\frac{r+1}{2}\right)}{\sqrt{\pi}}  A^r_{\lambda, \alpha} \ell^r_{\lambda + \frac{2}{\alpha}} < \infty,
    \]
    which ends the proof.
\end{proof}

\subsection{Approximation of the volatility}

Next, consider a family of sandwiched process $Y_m = \{Y_m(t),~t\in[0,T]\}$, $m\ge 1$, defined by equations of the form \eqref{eq: sandwiched approximation}. Note that the conditions of Theorem \ref{th: properties of Y} are met for all $m\ge 1$, i.e. each $Y_m$ is well-defined. Since constants $L_1$, $L_2$ in \eqref{eq: upper and lower bounds for sandwiched volatility, general case} depend only on $Y(0)$, the shape of $b$ and $\lambda \in \left(\frac{1}{1+\gamma}, H\right)$ that can be chosen jointly for all $m\ge 1$, we have that
\begin{equation}\label{eq: upper and lower bounds for sandwiched volatility, approximations}
    \varphi(t) + \frac{L_1}{(L_2 + \Lambda_m)^{\frac{1}{\gamma \lambda + \lambda - 1}}} < Y_m(t) < \psi(t) - \frac{L_1}{(L_2 + \Lambda_m)^{\frac{1}{\gamma \lambda + \lambda - 1}}}, \quad t\in [0,T], \quad m\ge 1.
\end{equation}
Moreover, the following result is true.

\begin{lemma}\label{lemma: Holder continuity of approximated volatilities}
    Let Assumptions \ref{assum: kernel}, \ref{assum: drift} and \ref{assum: approx kernels} hold. Then, for all $\lambda\in \left(0, H\right)$, there exists a sequence of positive random variables $\{\Upsilon_m,~m\ge 1\}$ such that $\sup_{m\ge 1} \mathbb E[\Upsilon^r_m] < \infty$ for any $r>0$ and
    \[
        |Y_m(t) - Y_m(s)| \le \Upsilon_m |t-s|^\lambda, \quad t,s\in[0,T].
    \]
\end{lemma}
 
\begin{proof}
    It is enough to prove the claim for any fixed $\lambda\in \left(\frac{1}{1+\gamma}, H\right)$. Taking the corresponding $\Lambda_m$ from Lemma \ref{lemma: Holder continuity of approximated noises} and proceeding exactly as in the proof of \cite[Lemma 3.6]{DNMYT2022-1}, one can see that
    \[
        |Y_m(t) - Y_m(s)| \le C\left(1+\Lambda_m + \frac{(L_2 + \Lambda_m)^{\frac{p}{\gamma \lambda + \lambda - 1}}}{L_1^p}\right)|t-s|^\lambda,
    \]
    where $C>0$ is a deterministic constant that does not depend on $m$, $p$ is from Assumption \textbf{(Y2)} and $L_1$, $L_2$ are from \eqref{eq: upper and lower bounds for sandwiched volatility, approximations}. Putting $\Upsilon_m:= C\left(1+\Lambda_m + \frac{(L_2 + \Lambda_m)^{\frac{p}{\gamma \lambda + \lambda - 1}}}{L_1^p}\right)$ and taking into account \eqref{eq: uniform bound on Holder seminorms of the noise}, it is easy to see that $\sup_{m\ge 1} \mathbb E[\Upsilon^r_m] < \infty$ for any $r>0$ as required.
\end{proof}

The next result gives a pathwise estimate of distance between $Y$ and $Y_m$.

\begin{theorem}\label{th: approximation of volatility}
    Let the processes $Y$ and $Y_m$, $m\ge1$, satisfy Assumptions \ref{assum: kernel}, \ref{assum: drift} and \ref{assum: approx kernels}. Then, for each $m\ge 1$, there exist a deterministic constant $C>0$ that does not depend on $m$ or $t$ and a random variable $\xi_m$ that does not depend on $t$ such that, for any $r>0$,
    \begin{equation}\label{eq: xim uniformly bounded moments}
        \sup_{m\ge 1} \mathbb E [\xi_m^r] < \infty
    \end{equation}
    and, for any $t\in[0,T]$,
    \[
        |Y(t) - Y_m(t)| \le C\left( |Z(t) - Z_m(t)| + \xi_m\int_0^t |Z(u) - Z_m(u)| du \right).
    \]
\end{theorem}
\noindent The proof of Theorem \ref{th: approximation of volatility} is rather long and technical, so we present it in Appendix \ref{appendix: proof of approximation of volatility}.


\subsection{Approximation of the price}

We are now ready to study the error estimates for the approximations of the prices. For this, we first present a result on the moments of the approximating processes. Let us consider the price process $S$ given by \eqref{eq: S} and denote
\[
    W(t) := \rho B_1(t) + \sqrt{1 - \rho^2} B_2(t), \quad t\in [0,T].
\]
Note that Theorem \ref{th: properties of price} immediately implies that for each $m \ge 1$ and all $r\in\mathbb R$
\[
    \mathbb E \left[ \sup_{t\in[0,T]} S_m^r(t) \right] < \infty, \quad \mathbb E \left[ \sup_{t\in[0,T]} X_m^r(t) \right] < \infty.
\]
However, in the sequel we will require a slightly stronger result summarized in the following Lemma. 

\begin{lemma}\label{lemma: moments of prices}
    Under Assumptions \ref{assum: kernel}, \ref{assum: drift} and \ref{assum: approx kernels}, for any $r\in\mathbb R$
    \[
        \sup_{m\ge 1} \mathbb E\left[ \sup_{t\in[0,T]}S^r_m(t) \right] < \infty, \quad \sup_{m\ge 1} \mathbb E\left[ \sup_{t\in[0,T]}X^r_m(t) \right] < \infty.
    \]
\end{lemma}
\begin{proof}
    It is sufficient to prove the result for $S_m$. Using the same argument as in the proof of \cite[Theorem 2.6]{DNMYT2022}, it is straightforward to show that 
    \begin{align}\label{proofeq: uniform expectation for approx prices}
        \mathbb E\left[\sup_{t\in[0,T]}S_m^r(t)\right] &\le C\left( 1+C_1 \sqrt{T}\sup_{s\in[0,T]}\psi(s) \exp\left\{ \frac{r^2T}{2}\max_{s\in[0,T]}\psi^2(s) \right\}\right),
    \end{align}
    where
    \[
        C := S^r(0)\exp\left\{|r|T\max_{s\in[0,T]}|\nu(s)| + \frac{|r|(|r|+1)T}{2}\max_{s\in[0,T]}\psi^2(s)\right\}
    \]
    and the constant $C_1$ comes from the Burkholder-Davis-Gundy inequality. The result follows from the fact that the right-hand side of \eqref{proofeq: uniform expectation for approx prices} does not depend on $m$.
    
\end{proof}

We are now finally ready to proceed to the main result of the Subsection.

\begin{theorem}\label{th: approx with difference kernel}
    Let Assumptions \ref{assum: kernel}, \ref{assum: drift}, \ref{assum: approx kernels} hold and $r\ge 2$ be fixed.
    \begin{itemize}
        \item[1)] If $\mathcal K, \mathcal K_m \in L^r([0,T]^2)$, $m\ge 1$, then there exists a constant $C>0$ that does not depend on $m$ such that
        \[
            \mathbb E\left[\sup_{t\in[0,T]}|S(t) - S_m(t)|^r\right] \le C\lVert \mathcal K - \mathcal K_m \rVert^r_{L^r([0,T]^2)}, \quad \mathbb E\left[\sup_{t\in[0,T]}|X(t) - X_m(t)|^r\right] \le C\lVert \mathcal K - \mathcal K_m \rVert^r_{L^r([0,T]^2)}.
        \]
        \item[2)] If $\mathcal K, \mathcal K_m \in L^2([0,T]^2)$, $m\ge 1$, \emph{are of the difference type}, i.e.
        \[
        \mathcal K(t,s) = \mathcal K(t-s)\mathbbm 1_{s<t}, \quad \mathcal K_m(t,s) = \mathcal K_m(t-s)\mathbbm 1_{s<t}, \quad t,s \in [0,T],
        \]
        then there exists constant $C>0$ that does not depend on $m$ such that
        \[
            \mathbb E\left[\sup_{t\in[0,T]}|S(t) - S_m(t)|^r\right] \le C \lVert \mathcal K - \mathcal K_m \rVert^r_{L^2([0,T])}, \quad \mathbb E\left[\sup_{t\in[0,T]}|X(t) - X_m(t)|^r\right] \le C \lVert \mathcal K - \mathcal K_m \rVert^r_{L^2([0,T])}.
        \]
    \end{itemize}
\end{theorem}

\begin{proof} It is sufficient to prove both claims for the non-discounted price $S$.

\noindent    \textit{Item 1.} For any $t\in[0,T]$, \eqref{eq: S}, \eqref{eq: price dynamics approximated} and the continuity of $\nu$ imply
    \begin{align*}
        |S(t) - S_m(t)|^r &\le C \left(\int_0^t |S(s) - S_m(s)|^r ds + \left|\int_0^t (Y(s)S(s) - Y_m(s) S_m(s))dW(s)\right|^r \right) 
        \\
        &\le C \int_0^t \sup_{u\in[0,s]}|S(u) - S_m(u)|^r ds  + C \left(\sup_{u\in[0,t]}\left|\int_0^u (Y(s)S(s) - Y_m(s) S_m(s))dW(s)\right|\right)^r , 
    \end{align*}
    whence
    \begin{align*}
        \sup_{u\in[0,t]}|S(u) - S_m(u)|^r &\le C \int_0^t \sup_{u\in[0,s]}|S(u) - S_m(u)|^r ds 
        \\
        &\quad + C \left(\sup_{u\in[0,t]}\left|\int_0^u (Y(s)S(s) - Y_m(s) S_m(s))dW(s)\right|\right)^r.
    \end{align*}
    By the Burkholder-Davis-Gundy inequality,
    \begin{equation}\label{eq: pre-Gron diff of S and Sm 1}
    \begin{aligned}
        \mathbb E&\left[\sup_{u\in[0,t]}|S(u) - S_m(u)|^r\right] 
        \\
        &\le C \int_0^t \mathbb E\left[\sup_{u\in[0,s]}|S(u) - S_m(u)|^r\right]ds + C \mathbb E\left[ \left(\int_0^t \big(Y(s)S(s) - Y_m(s) S_m(s)\big)^2 ds \right)^{\frac{r}{2}}\right]
        \\
        & \le C \int_0^t \mathbb E\left[\sup_{u\in[0,s]}|S(u) - S_m(u)|^r\right]ds  + C  \int_0^t \mathbb E\left[\big| Y(s)S(s) - Y_m(s) S_m(s)\big|^r\right] ds.
    \end{aligned}    
    \end{equation}
    Let us focus on the last summand of \eqref{eq: pre-Gron diff of S and Sm 1}. Using the fact that
    \[
        \max_{m\ge 1} \max_{t\in[0,T]} |Y_m(t)| \le \max_{t\in[0,T]}|\psi(t)|,
    \]
    we have
    \begin{equation}\label{eq: pre-Gron diff of S and Sm 2}
    \begin{aligned}
        \int_0^t &\mathbb E\left[ |Y(s)S(s) - Y_m(s) S_m(s)|^r \right] ds 
        \\
        &\le C\int_0^t \mathbb E \left[ Y^r_m(s)|S(s) - S_m(s)|^r \right]ds + C\int_0^t \mathbb E\left[ S^r(s)|Y(s) - Y_m(s)|^r\right]ds
        \\
        &\le C\int_0^t \mathbb E \left[ \sup_{u\in[0,s]}|S(u) - S_m(u)|^r \right]ds  + C\int_0^t \mathbb E\left[ S^r(s)|Y(s) - Y_m(s)|^r\right]ds.
    \end{aligned}    
    \end{equation}
    Studying again the last summand of \eqref{eq: pre-Gron diff of S and Sm 2}, we obtain by Theorem \ref{th: approximation of volatility} that there exist random variables $\xi_m$, $m\ge 1$, such that \eqref{eq: xim uniformly bounded moments} holds and
    \begin{equation}\label{eq: pre-Gron diff of S and Sm 3}
    \begin{aligned}
        \int_0^t &\mathbb E\left[ S^r(s)|Y(s) - Y_m(s)|^r \right]ds 
        \\
        & \le C\int_0^t \mathbb E\left[ S^r(s)|Z(s) - Z_m(s)|^r\right]ds + C\int_0^t \mathbb E\left[ S^r(s) \xi^r_m \int_0^s |Z(u) - Z_m(u)|^r du \right]ds.
    \end{aligned}    
    \end{equation}
    Also, by \cite[{Theorem 2.6}]{DNMYT2022},
    \[
         \mathbb E \left[\sup_{t\in[0,T]} S^{2r}(t) \right] < \infty,
    \]
    so the H\"older inequality as well as the Gaussianity of the random variable $Z(u) - Z_m(u)$ yield
    \begin{equation}\label{proofeq: price error 1} 
    \begin{aligned}
        \int_0^t &\mathbb E\left[ S^r(s)|Z(s) - Z_m(s)|^r\right]ds \le \int_0^t \left(\mathbb E\left[ S^{2r}(s)\right]\right)^{\frac{1}{2}}\left(\mathbb E\left[|Z(s) - Z_m(s)|^{2r}\right]\right)^{\frac{1}{2}}ds
        \\
        &\le C \int_0^t \left(\int_0^s (\mathcal K(s,u) - \mathcal K_m(s,u))^2 du\right)^{\frac{r}{2}} ds \le C\lVert \mathcal K - \mathcal K_m \rVert^r_{L^r([0,T]^2)}.
    \end{aligned}
    \end{equation}
    Additionally, \eqref{eq: xim uniformly bounded moments} together with \cite[{Theorem 2.6}]{DNMYT2022} imply that
    \[
        \sup_{m\ge 1} \mathbb E [\xi^{2r}_m\sup_{t\in [0,T]} S^{2r}(t) ] < \infty  
    \]
    hence
    \begin{equation}\label{proofeq: price error 2} 
    \begin{aligned}
        \int_0^t &\mathbb E\left[ S^r(s) \xi^r_m \int_0^s |Z(u) - Z_m(u)|^r du \right]ds 
        \\
        &\le \int_0^t \int_0^s\left(\mathbb E\left[ S^{2r}(s) \xi^{2r}_m \right]\right)^{\frac{1}{2}} \left(\mathbb E \left[ |Z(u) - Z_m(u)|^{2r}  \right] \right)^{\frac{1}{2}}duds
        \\
        &\le C \int_0^t \int_0^s\left(\mathbb E\left[|Z(u) - Z_m(u)|^{2r}\right] \right)^{\frac{1}{2}} duds \le C \int_0^t \int_0^s \left(\int_0^u \left(\mathcal K(u,v) - \mathcal K_m(u,v)\right)^2 dv\right)^{\frac{r}{2}} duds
        \\
        &\le C \lVert \mathcal K - \mathcal K_m \rVert^r_{L^r([0,T]^2)}.
    \end{aligned}
    \end{equation}
    Therefore, taking into account \eqref{eq: pre-Gron diff of S and Sm 1}--\eqref{proofeq: price error 2}, we see that there exists a constant $C>0$ that does not depend on $m$ or the particular choice of $t\in[0,T]$ such that
    \begin{equation}\label{proofeq: price error 3}
    \begin{aligned}
        \mathbb E&\left[\sup_{u\in[0,t]}|S(u) - S_m(u)|^r\right] \le C \int_0^t \mathbb E\left[\sup_{u\in[0,s]}|S(u) - S_m(u)|^r\right]ds + C \lVert \mathcal K - \mathcal K_m \rVert^r_{L^r([0,T]^2)},
    \end{aligned}  
    \end{equation}
    and item 1) now follows from the Gronwall's inequality.
    
\noindent     \textit{Item 2.} If $\mathcal K$ and $\mathcal K_m$ both have the form
    \[
        \mathcal K(t,s) = \mathcal K(t-s)\mathbbm 1_{s<t}, \quad \mathcal K_m(t,s) = \mathcal K_m(t-s)\mathbbm 1_{s<t}, \quad t,s \in [0,T],
    \]
    then \eqref{proofeq: price error 1}--\eqref{proofeq: price error 2} can be rewritten as
    \begin{align*}
        \int_0^t \mathbb E\left[ S^r(s)|Z(s) - Z_m(s)|^r\right]ds & \le C \int_0^t \left(\int_0^s (\mathcal K(s-u) - \mathcal K_m(s-u))^2 du\right)^{\frac{r}{2}} ds
        \\
        & \le  C \int_0^t \lVert \mathcal K - \mathcal K_m \rVert^r_{L^2([0,T])} ds \le C\lVert \mathcal K - \mathcal K_m \rVert^r_{L^2([0,T])}
    \end{align*}
    and
    \begin{align*}
        \int_0^t &\mathbb E\left[ S^r(s) \xi^r_m \int_0^s |Z(u) - Z_m(u)|^r du \right]ds \le C \int_0^t \int_0^s \left(\int_0^u \left(\mathcal K(u-v) - \mathcal K_m(u-v)\right)^2 dv\right)^{\frac{r}{2}} duds
        \\
        &\le C \int_0^t \int_0^s \lVert \mathcal K - \mathcal K_m \rVert^r_{L^2([0,T])} duds \le C\lVert \mathcal K - \mathcal K_m \rVert^r_{L^2([0,T])}.
    \end{align*}
    Whence \eqref{proofeq: price error 3} becomes
    \begin{equation*}
    \begin{aligned}
        \mathbb E&\left[\sup_{u\in[0,t]}|S(u) - S_m(u)|^r\right] \le C \int_0^t \mathbb E\left[\sup_{u\in[0,s]}|S(u) - S_m(u)|^r\right]ds + C \lVert \mathcal K - \mathcal K_m \rVert^r_{L^2([0,T])}.
    \end{aligned}    
    \end{equation*}
    The required result then again follows from Gronwall's inequality.
\end{proof}

\section{Mean-variance hedging in the SVV model}\label{s: NA-derivative estimates}

In this Section, we will deal with the \textit{mean-variance hedging} problem \eqref{eq: J} withing the SVV model \eqref{eq: S}--\eqref{eq: Y}, where the payoff $F$ is of European type, i.e. $F = f(X(T))$. Note that Theorem \ref{th: properties of price} guarantees that $F = f(X(T)) \in L^2(\mathbb P)$ as long as the payoff function is of polynomial growth, i.e. the problem \eqref{eq: J} is well-posed and has a solution given by the Galtchouk-Kunita-Watanabe decomposition. In this situation, the optimal hedging portfolio coincides with the \textit{non-anticipating derivative} of $F$ with respect to the martingale $X$ the numerical computation of which heavily benefits from the Markovian approximation derived in Section~\ref{s: Markovian approximation}.

This Section is organized as follows. In Subsection~\ref{subsec: NA-derivative}, we provide all the necessary details on the non-anticipating derivatives and their connection to the optimization problem \eqref{eq: J}. In Subsection~\ref{subsec: Markovian approximation of the NA-derivative}, we estimate the error between the solution to the original hedging problem \eqref{eq: J} and its Markovian counterpart. The numerical algorithms for computing optimal hedging strategies are described in the subsequent Section~\ref{s: Monte Carlo approximation of the hedging strategy}.

\subsection{Non-anticipating derivative and mean-variance hedging}\label{subsec: NA-derivative}

Let $\xi \in L^2(\mathbb P)$ and $\eta = \{\eta(t),~t\in[0,T]\}$ be a square-integrable martingale w.r.t. a filtration $\mathbb G = \{\mathcal G_t,~t\in[0,T]\}$. For an arbitrary partition $\pi = \{0=t_0 < t_1 < ... < t_n = T\}$ of $[0,T]$ with the mesh $|\pi| := \max_{k} (t_k - t_{k-1})$, denote
\begin{equation}\label{eq: NA-derivative pre-limit}
\begin{aligned}
    u_{\pi} &:= \sum_{k=0}^{n-1} u_{\pi,k} \mathbbm 1_{(t_k, t_{k+1}]}, \quad u_{\pi,k} &:=  \frac{\mathbb E \left[(\eta(t_{k+1}) - \eta(t_k))\xi~|~\mathcal G_{t_k}\right]}{\mathbb E\left[ (\eta(t_{k+1}) - \eta(t_k))^2~|~\mathcal G_{t_k}\right]}, \quad k = 0, 1, ..., n-1.
\end{aligned}
\end{equation}

\begin{definition}\label{def: NA-derivative}
    Consider a monotone sequence of partitions $\{\pi_M,~M\ge 1\}$, such that $|\pi_M| \to 0$ as $M\to \infty$. The $L^2(\mathbb P \times [\eta])$-limit 
    \begin{equation}\label{eq: NA derivative integral sums}
        \mathfrak D \xi := \lim_{|\pi_M| \to 0} u_{\pi_M}
    \end{equation}
     is called the \textit{non-anticipating derivative} of $\xi$ w.r.t. $\eta$.
\end{definition}

\noindent It turns out that $\mathfrak D \xi$ is well-defined and $\int_0^T \mathfrak D \xi(s) d\eta(s)$ gives the orthogonal $L^2(\mathbb P \times [\eta])$-projection of $\xi$ on the subspace of stochastic integrals w.r.t. $\eta$. 

\begin{theorem}{\textit{(\cite[Theorem 2.1]{DiNunno})}}\label{th: NA-derivative properties}
    The non-anticipating derivative $\mathfrak D \xi$ is well-defined, i.e. the limit \eqref{eq: NA derivative integral sums} exists and does not depend on the particular choice of the partitions. Moreover, any $\xi \in L^2(\mathbb P)$ admits a unique representation of the form
    \[
        \xi = \xi_0 + \int_0^T \mathfrak D \xi(s)d\eta(s),
    \]
    where $\xi_0 \in L^2(\mathbb P)$ is such that $\mathfrak D \xi_0 = 0$ and $\mathbb E\left[ \xi_0 \int_0^T \mathfrak D \xi(s)d\eta(s) \right] = 0$. In other words, the infimum over all $\mathbb G$-adapted $\eta$-integrable strategies
    \[
        \inf_u \mathbb E\left[\left( \xi - \int_0^T u(s) d\eta(s) \right)^2\right]
    \]
    is attained at $u = \mathfrak D \xi$.
\end{theorem}

\begin{remark}
    By Theorem \ref{th: NA-derivative properties}, it is easy to see that the linear operator $\mathfrak D$ is actually the dual of the It\^o integral.
\end{remark}

\begin{notation*}
    Since the limit \eqref{eq: NA derivative integral sums} does not depend on the particular choice of partition $\{\pi_M,~M\ge 1\}$, we will use the notation 
    \[
        \mathfrak D \xi := \lim_{|\pi| \to 0} u_{\pi},
    \]
    meaning that the limit in $L^2(\mathbb P \times [\eta])$ is taken along an arbitrary sequence of monotone partitions with the mesh going to zero.
\end{notation*}

\noindent Next, let Assumptions \ref{assum: kernel} and \ref{assum: drift} hold, $S$ and $X$ be defined by \eqref{eq: S} and \eqref{eq: X} respectively. Consider the mean-square hedging problem \eqref{eq: J}, where the payoff function $f$: $\mathbb R_+ \to \mathbb R$ satisfies the following assumption.

\begin{assum}{(F)}\label{assum: F}
    The payoff function $f$: $\mathbb R_+ \to \mathbb R$ can be represented as $f = f_1 + f_2$, where
    \begin{itemize}
        \item[(i)] $f_1$ is globally Lipschitz, i.e. there exists $C > 0$ such that
        \[
            |f_1(t) - f_1(s)| \le C|t-s|, \quad s,t \ge 0;
        \]
        \item[(ii)] $f_2$: $\mathbb R_+ \to \mathbb R$ is of bounded variation over $\mathbb R_+$, i.e.
        \[
            V(f_2) := \lim_{x\to\infty} V_{[0,x]}(f_2) < \infty,
        \]
        where 
        \[
            V_{[0,x]}(f_2) := \sup \sum_{j=1}^N |f_2(x_j) - f_2(x_{j-1})|
        \]
        and the supremum is taken over all $N\ge 0$ and all partitions $0=x_0 < x_1 < ... < x_N = x$.
    \end{itemize}
\end{assum}

It is evident that Eq. \eqref{eq: sup of S and X} and Assumption \ref{assum: F} imply that $F = f(X(T))\in L^2(\mathbb P)$, i.e. Theorem \ref{th: NA-derivative properties} holds and we have the following corollary.

\begin{corollary}
    The non-anticipating derivative $u = \mathfrak D F$ of $F$ w.r.t. $X$ is the minimizer of $J$ from \eqref{eq: J} and hence represents the optimal hedging portfolio.
\end{corollary}

\begin{remark}\hfill
     Consider an arbitrary partition $\pi = \{0=t_0 < t_1 < ... < t_n = T\}$. The proof of \cite[Theorem 2.1]{DiNunno} implies that the pre-limit sum
        \begin{align}\label{eq: NA-derivative pre-limit X}
            u_{\pi} := \sum_{k=0}^{n-1} u_{\pi,k} \mathbbm 1_{(t_k, t_{k+1}]}, \quad u_{\pi,k} :=  \frac{\mathbb E \left[(X(t_{k+1}) - X(t_k))F~|~\mathcal F_{t_k}\right]}{\mathbb E\left[ (X(t_{k+1}) - X(t_k))^2~|~\mathcal F_{t_k}\right]}, \quad k = 0, 1, ..., n-1,
        \end{align}
        is the $L^2(\mathbb P \times [\eta])$-orthogonal projection of $F = f(X(T))$ onto the subspace generated by stochastic integrals of simple processes
        \begin{equation}\label{eq: simple processes paper 5}
            \sum_{k=0}^{n-1} \zeta_k \mathbbm 1_{(t_k, t_{k+1}]}
        \end{equation}
         w.r.t. $X$. Note that admissible portfolios in real markets are exactly of this type since there is no technical possibility of real ``continuous'' trading.
\end{remark}

\begin{remark}
    Note that formula \eqref{eq: NA-derivative pre-limit X} is explicit in the sense that the hedge is written only in terms of the discounted price model, the information flow of reference, and the claim $F$. This formula is in the spirit of the Clark-Haussmann-Ocone (CHO) formula (see e.g. \cite{Nualart_2006, Nunno_Oksendal_Proske_2009}) but we stress that the non-anticipating derivative has several important advantages. Namely, the CHO formula exploits the Malliavin derivative which is tailored for specific noises (e.g. Brownian motion) and claims $F$ falling in the domain of the Malliavin operator. In turn, formula \eqref{eq: NA-derivative pre-limit X} is available for all square integrable claims and all square integrable martingales as discounted prices.
\end{remark} 

\subsection{Non-anticipating derivative and Markovian approximations}\label{subsec: Markovian approximation of the NA-derivative}

The very definition of the non-anticipating derivative $u = \mathfrak D F$ describing the hedging portfolio provides a natural approximation of it. Indeed, simple processes $u_\pi$ given by \eqref{eq: NA-derivative pre-limit X} converge in $L^2(\mathbb P \times [X])$ to $\mathfrak D F$ and each $u_\pi$ is itself the optimal hedge in the corresponding class of simple processes \eqref{eq: simple processes paper 5}. However, the \textit{computation} of the conditional expectations in \eqref{eq: NA-derivative pre-limit X} is a challenging task that becomes even more complicated since the Volterra noise $Z$ from \eqref{eq: Y} may have memory (and thus, in general, $X$ is not Markovian).

In what follows, we will utilize the Markovian approximation derived in Section \ref{s: Markovian approximation} and compare the non-anticipating derivative $\mathfrak D f(X(T))$ w.r.t. $X$ with $\mathfrak D f(X_m(T))$ w.r.t. $X_m$. Note that we are interested in European options with rather complicated payoffs covering the ones with discontinuities, so we will need to compare $F = f(X(T))$ with $F_m := f(X_m(T))$. In order to do that, the following Theorem will be used.

\begin{theorem} \label{th: error in payoff}
    Let $f = f_1 + f_2$ satisfy Assumption \ref{assum: F} and $\xi$, $\hat\xi \in L^r(\mathbb P)$, $1\le r < \infty$. Then
    \begin{itemize}
        \item[1)] there exists a constant $C>0$ depending only on $f_1$ and $r$ such that 
        \[
            \mathbb E [|f_1(\xi) - f_1(\hat\xi)|^r] \le C \mathbb E[|\xi - \hat\xi|^{r}];
        \]
        \item[2)] if $\xi$ has bounded density $\phi_\xi$, then
        \begin{equation}\label{eq: error in payoff discontinuous}
            \mathbb E [|f_2(\xi) - f_2(\hat\xi)|^r] \le 3^{r+1} V^r(f_2) (\sup \phi_{\xi})^{\frac{r}{r+1}} \left(\mathbb E\left[ |\xi - \hat \xi|^r\right]\right)^{\frac{r}{r+1}};
        \end{equation}
        \item[3)] if $\xi$ has bounded density $\phi_\xi$, then there exists a constant $C>0$ depending only on $f_1$, $f_2$, $r$ and the density $\phi_\xi$ such that
        \[
            \mathbb E [|f(\xi) - f(\hat\xi)|^r] \le C \left(\mathbb E[|\xi - \hat\xi|^{r}] + \left(\mathbb E\left[ |\xi - \hat \xi|^r\right]\right)^{\frac{r}{r+1}}\right).
        \]
    \end{itemize}
\end{theorem}
\begin{proof}
    Item 1) is obvious and follows directly from the Lipschitz condition for $f_1$. Item 2) is proved in \cite[Theorem 2.4]{Avikainen2009}. Item 3) is a combination of 1) and 2).
\end{proof}

\begin{assumption}
    Throughout this Subsection, we always assume that Assumptions \ref{assum: kernel} and \ref{assum: drift} hold, the sequence of kernels $\{\mathcal K_m,~m\ge 1\}$ satisfies Assumption \ref{assum: approx kernels} and $\mathcal K$, $\mathcal K_m$, $m\ge 1$, are all difference kernels, i.e.
    \[
        \mathcal K(t,s) = \mathcal K(t-s)\mathbbm 1_{s<t}, \quad \mathcal K_m(t,s) = \mathcal K_m(t-s)\mathbbm 1_{s<t}, \quad t,s \in [0,T].
    \]
    Moreover, we also impose Assumption \ref{assum: F} on the payoff function $f = f_1 + f_2$.
\end{assumption}
\noindent To allow for compact writing, denote
\[
    W(t) := \rho B_1(t) + \sqrt{1 - \rho^2} B_2(t), \quad t\in [0,T],
\]
where $\rho \in (-1,1)$ is from \eqref{eq: S}. Consider $F_m := f(X_m(T))$ and observe that for any $r\in\mathbb R$
\begin{equation}\label{eq: uniform bound on moments of payoff}
    \mathbb E \left[ F^r \right] + \sup_{m\ge 0} \mathbb E \left[ F_m^r \right] < \infty
\end{equation}
by Lemma \ref{lemma: moments of prices}. For a given partition $\pi = \{0 = t_0 < t_1 < ... < t_n = T\}$, denote also $\Delta X(t_k) := X(t_{k+1}) - X(t_{k})$, $\Delta X_m(t_k) := X_m(t_{k+1}) - X_m(t_{k})$ and consider the non-anticipating derivatives
\begin{align*}
    \mathfrak D F  := \lim_{|\pi| \to 0} u_{\pi}, \quad \mathfrak D F_m  := \lim_{|\pi| \to 0} u^m_{\pi},
\end{align*}
where
\begin{align}
    u_{\pi} := \sum_{k=0}^{n-1} u_{\pi,k} \mathbbm 1_{(t_k, t_{k+1}]}, \quad &u_{\pi,k} :=  \frac{\mathbb E \left[F\Delta X(t_k)~|~\mathcal F_{t_k}\right]}{\mathbb E\left[ (\Delta X(t_k))^2~|~\mathcal F_{t_k}\right]},
    ,\label{eq: definition of u}
    \\
    u^m_{\pi} := \sum_{k=0}^{n-1} u^m_{\pi,k} \mathbbm 1_{(t_k, t_{k+1}]},\quad &u^m_{\pi,k} :=  \frac{\mathbb E \left[F_m\Delta X_m(t_k)~|~\mathcal F_{t_k}\right]}{\mathbb E\left[ (\Delta X_m(t_k))^2~|~\mathcal F_{t_k}\right]}.\label{eq: definition of um}
\end{align}

\begin{lemma}\label{lemma: difference betweem u and um}
    Let Assumptions \ref{assum: kernel}, \ref{assum: drift} and \ref{assum: approx kernels} hold, the payoff function $F$ satisfy Assumption \ref{assum: F} and both $\mathcal K$ and $\mathcal K_m$, $m\ge 1$, have the form
    \[
        \mathcal K(t,s) = \mathcal K(t-s)\mathbbm 1_{s<t}, \quad \mathcal K_m(t,s) = \mathcal K_m(t-s)\mathbbm 1_{s<t}, \quad t,s \in [0,T].
    \]
    Then there exists a constant $C>0$ that does not depend on $m$ such that for any partition $\pi = \{0 = t_0 < t_1 < ... < t_n = T\}$
    \[
        \mathbb E [| u_{\pi, k} - u^m_{\pi, k} |] \le \frac{C}{\sqrt{t_{k+1} - t_k}}\left( \left(\mathbb E\left[ (F-F_m)^4 \right]\right)^{\frac{1}{4}} + \lVert \mathcal K - \mathcal K_m \rVert_{L^2([0,T])}\right), \quad k=0,...,n-1,
    \]
    where $u_{\pi,k}$ and $u^m_{\pi,k}$ are defined in \eqref{eq: definition of u} and \eqref{eq: definition of um}.
\end{lemma}
\begin{proof}
    Fix $k=0,...,n-1$ and denote $\Delta := t_{k+1} - t_k$. It is easy to see that
    \begin{equation}\label{proofeq: approximation of summands in NA-derivative 1}
    \begin{aligned}
        \mathbb E [| u_{\pi, k} - u^m_{\pi, k} |] &\le \mathbb E \left[ \frac{|F - F_m| |\Delta X(t_k)| }{\mathbb E \left[ (\Delta X(t_k))^2~|~\mathcal F_{t_k} \right]} \right] + \mathbb E \left[ |F_m| \frac{\left| \Delta X(t_k) - \Delta X_m(t_k) \right|}{ \mathbb E \left[ (\Delta X(t_k))^2~|~\mathcal F_{t_k} \right] } \right]
        \\
        &\quad + \mathbb E \left[ |F_m| \left| \Delta X_m(t_k) \right|\frac{\mathbb E \left[ |(\Delta X_m(t_k))^2  -  (\Delta X(t_k))^2|~|~\mathcal F_{t_k} \right]}{ \mathbb E \left[ (\Delta X(t_k))^2~|~\mathcal F_{t_k} \right] \mathbb E \left[ (\Delta X_m(t_k))^2~|~\mathcal F_{t_k} \right]} \right]
        \\
        &=: I_1 + I_2 + I_3.
    \end{aligned}
    \end{equation}
    Now we will deal with each of the above summands separately.
    
    \noindent \textbf{Step 1: estimation of} $\boldsymbol{I_1}.$ By Lemma \ref{lemma: bounds for conditional increments}, there exist constants $C_1$, $C_2$ (not depending on the partition or $m$) such that
    \[
        \mathbb E \left[ (\Delta X(t_k))^2~|~\mathcal F_{t_k} \right] \ge \Delta C_1 X^2(t_k)
    \]
    and
    \[
        \mathbb E \left[ (\Delta X(t_k))^2 \right] \le \Delta C_2 \mathbb E\left[X^2(t_k)\right].
    \]
    Hence, using H\"older inequality and \eqref{eq: sup of S and X}, one can write
    \begin{equation}\label{proofeq: approximation of summands in NA-derivative 2}
    \begin{aligned}
        I_1 &:= \mathbb E \left[ \frac{|F - F_m| |\Delta X(t_k)| }{\mathbb E \left[ (\Delta X(t_k))^2~|~\mathcal F_{t_k} \right]} \right] \le \frac{C}{\Delta} \mathbb E \left[ |F - F_m| |\Delta X(t_k)| X^{-2}(t_k) \right]
        \\
        &\le \frac{C}{\Delta} \left(\mathbb E\left[(\Delta X(t_k))^2\right]\right)^{\frac{1}{2}}  \left(\mathbb E [X^{-8}(t_k)]\right)^{\frac{1}{4}} \left( \mathbb E \left[ |F - F_m|^4 \right] \right)^{\frac{1}{4}}
        \\
        &\le \frac{C}{\sqrt{\Delta}} \left(\mathbb E\left[ X^2(t_k)\right]\right)^{\frac{1}{2}} \left( \mathbb E \left[ |F - F_m|^4 \right] \right)^{\frac{1}{4}}\le \frac{C}{\sqrt{\Delta}} \left( \mathbb E \left[ |F - F_m|^4 \right] \right)^{\frac{1}{4}},
    \end{aligned}
    \end{equation}
    where $C>0$ can be chosen to not depend on $m$ or the partition.
    
    \noindent \textbf{Step 2: estimation of} $\boldsymbol{I_2}.$ By H\"older inequality, Lemma \ref{lemma: bounds for conditional increments}, Lemma \ref{lemma: difference between increments} and \eqref{eq: uniform bound on moments of payoff},
    \begin{equation}\label{proofeq: approximation of summands in NA-derivative 3}
    \begin{aligned}
        I_2 := \mathbb E \left[ |F_m| \frac{\left| \Delta X(t_k) - \Delta X_m(t_k) \right|}{ \mathbb E \left[ (\Delta X(t_k))^2~|~\mathcal F_{t_k} \right] } \right] &\le \frac{C}{\Delta}\left(\mathbb E[ \left| \Delta X(t_k) - \Delta X_m(t_k) \right|^2 ]\right)^{\frac{1}{2}} \left( \mathbb E\left[F^2_m X^{-4}(t_k) \right] \right)^{\frac{1}{2}} 
        \\
        &\le \frac{C}{\sqrt{\Delta}}\lVert \mathcal K - \mathcal K_m \rVert^r_{L^2([0,T])},
    \end{aligned}    
    \end{equation}
    where $C$ does not depend on $m$ or the partition.

    \noindent \textbf{Step 3: estimation of} $\boldsymbol{I_3}.$ Using H\"older's inequality and \eqref{eq: uniform bound on moments of payoff}, one can verify that
    \begin{equation}\label{proofeq: approximation of summands in NA-derivative 4}
    \begin{aligned}
        I_3 &:= \mathbb E \left[ |F_m| \left| \Delta X_m(t_k) \right|\frac{\mathbb E \left[ |(\Delta X_m(t_k))^2  -  (\Delta X(t_k))^2|~|~\mathcal F_{t_k} \right]}{ \mathbb E \left[ (\Delta X(t_k))^2~|~\mathcal F_{t_k} \right] \mathbb E \left[ (\Delta X_m(t_k))^2~|~\mathcal F_{t_k} \right]} \right]
        \\
        &\le \frac{C}{\Delta^2} \left( \mathbb E\left[F_m^4\right] \right)^{\frac{1}{4}} \left(\mathbb E\left[ |\Delta X_m(t_k)|^4 \right]\right)^{\frac{1}{4}} \left(\mathbb E \left[ |(\Delta X_m(t_k))^2  -  (\Delta X(t_k))^2|^2 \right]\right)^{\frac{1}{2}}
        \\
        &\le \frac{C}{\Delta^2} \left(\mathbb E\left[ |\Delta X_m(t_k)|^4 \right]\right)^{\frac{1}{4}} \left(\mathbb E \left[ |(\Delta X_m(t_k))^2  -  (\Delta X(t_k))^2|^2 \right]\right)^{\frac{1}{2}}.
    \end{aligned}    
    \end{equation}
    By the Burkholder--Davis--Gundy and H\"older inequalities, uniform boundedness of $Y_m$ and \eqref{eq: sup of S and X},
    \begin{equation}\label{proofeq: approximation of summands in NA-derivative 5}
    \begin{aligned}
        \left(\mathbb E\left[ |\Delta X_m(t_k)|^4 \right]\right)^{\frac{1}{4}} &= \left(\mathbb E\left[ \left|\int_{t_k}^{t_{k+1}} Y_m(s)X_m(s) dW(s) \right|^4 \right]\right)^{\frac{1}{4}} \le \left( \mathbb E\left[ \left( \int_{t_k}^{t_{k+1}} Y^2_m(s) X^2_m(s) ds \right)^2 \right] \right)^{\frac{1}{4}}
        \\
        &\le \left( \Delta  \int_{t_k}^{t_{k+1}} \mathbb E\left[Y^4_m(s) X^4_m(s) \right]ds   \right)^{\frac{1}{4}} \le C \sqrt{\Delta} \left(\mathbb E\left[ \sup_{u\in [0,T]} X^4_m(u) \right]\right)^{\frac{1}{4}} \le C \sqrt{\Delta}.
    \end{aligned}
    \end{equation}
    Further, one can verify using It\^o formula that
    \begin{equation}\label{proofeq: nasty estimation after Ito}
    \begin{aligned}
        \mathbb E &\left[ |(\Delta X_m(t_k))^2  -  (\Delta X(t_k))^2|^2 \right] \le C \mathbb E\left[\left( \int_{t_k}^{t_{k+1}} \Big( Y^2_m(s)X^2_m(s) - Y^2(s) X^2(s)\Big) ds \right)^2\right]
        \\
        &\quad + C \mathbb E\Bigg[\bigg( 2\int_{t_k}^{t_{k+1}} \Big((X_m(s) - X_m(t_k)) Y_m(s)X_m(s) - (X(s) - X(t_k)) Y(s)X(s)\Big) dW(s) \bigg)^2\Bigg].
    \end{aligned}
    \end{equation}
    We study the last two summands separately. It is clear that
    \begin{equation}\label{proofeq: technical estimation}
    \begin{aligned}
        \mathbb E&\left[\left( \int_{t_k}^{t_{k+1}} \Big( Y^2_m(s)X^2_m(s) - Y^2(s) X^2(s)\Big) ds \right)^2\right] \le \Delta  \int_{t_k}^{t_{k+1}} \mathbb E\left[\Big( Y^2_m(s)X^2_m(s) - Y^2(s) X^2(s)\Big)^2\right] ds
        \\
        &\le C \Delta  \int_{t_k}^{t_{k+1}} \mathbb E\left[\Big( X^2_m(s) - X^2(s)\Big)^2\right] ds + C \Delta  \int_{t_k}^{t_{k+1}} \mathbb E\left[X^2(s)\Big( Y^2_m(s) - Y^2(s)\Big)^2\right] ds.
    \end{aligned}    
    \end{equation}
    Taking into account that processes $\{X^2(t),~t\in[0,T]\}$ and $\{X_m^2(t),~t\in[0,T]\}$ satisfy the stochastic differential equations
    \begin{equation*}
    \begin{gathered}
        dX^2(t) = Y^2(t)X^2(t)dt + 2Y(t)X^2(t)dW(t),
        \\
        dX_m^2(t) = Y_m^2(t)X_m^2(t)dt + 2Y_m(t)X_m^2(t)dW(t),
    \end{gathered}    
    \end{equation*}
    boundedness of $Y$, uniform boundedness of $\{Y_m,~m\ge 1\}$ as well as the same argument as in Theorem \ref{th: approx with difference kernel} imply that there exists a constant $C>0$ such that
    \[
        \mathbb E\left[\sup_{s\in[0,T]}\Big( X^2_m(s) - X^2(s)\Big)^2\right] ds \le C\lVert \mathcal K - \mathcal K_m \rVert^2_{L^2([0,T])},
    \]
    whereas, again like in the proof of Theorem \ref{th: approx with difference kernel}, we have
    \begin{equation*}
    \begin{aligned}
        \mathbb E\left[X^2(s)\Big( Y^2_m(s) - Y^2(s)\Big)^2\right]  & = \mathbb E\left[X^2(s)( Y_m(s) + Y(s))^2 ( Y_m(s) - Y(s))^2\right]
        \\
        &\le C \mathbb E\left[X^2(s) ( Y_m(s) - Y(s))^2\right] \le C \lVert \mathcal K - \mathcal K_m \rVert^2_{L^2([0,T])}.
    \end{aligned}    
    \end{equation*}
    Plugging these estimates in \eqref{proofeq: technical estimation}, we obtain that
    \begin{equation*}
    \begin{aligned}
        \mathbb E\left[\left( \int_{t_k}^{t_{k+1}} \Big( Y^2_m(s)X^2_m(s) - Y^2(s) X^2(s)\Big) ds \right)^2\right] &\le C \Delta^2 \lVert \mathcal K - \mathcal K_m \rVert^2_{L^2([0,T])}.
    \end{aligned}    
    \end{equation*}
    Next, for the second summand of \eqref{proofeq: nasty estimation after Ito} we have
    \begingroup
    \allowdisplaybreaks
    \begin{equation}\label{proofeq: approximation of summands in NA-derivative 6}
    \begin{aligned}
        \mathbb E&\Bigg[\bigg( \int_{t_k}^{t_{k+1}} \Big((X_m(s) - X_m(t_k)) Y_m(s)X_m(s) - (X(s) - X(t_k)) Y(s)X(s)\Big) dW(s) \bigg)^2\Bigg]
        \\
        & = \mathbb E \Bigg[\int_{t_k}^{t_{k+1}} \Big((X_m(s) - X_m(t_k)) Y_m(s)X_m(s) - (X(s) - X(t_k)) Y(s)X(s)\Big)^2 ds \Bigg]
        \\
        & \le C \int_{t_k}^{t_{k+1}} \mathbb E \left[ (X(s) - X(t_k))^2 (Y_m(s)X_m(s) - Y(s)X(s))^2 \right]ds
        \\
        &\quad + C \int_{t_k}^{t_{k+1}} \mathbb E \left[ \Big( (X_m(s) - X_m(t_k)) - (X(s)-X(t_k))\Big)^2 Y^2_m(s)X^2_m(s) \right]ds
        \\
        &\le C \int_{t_k}^{t_{k+1}} \left(\mathbb E \left[ (X(s) - X(t_k))^4\right]\right)^{\frac{1}{2}} \left( \mathbb E\left[ (Y_m(s)X_m(s) - Y(s)X(s))^4 \right]\right)^{\frac{1}{2}}ds
        \\
        &\quad + C \int_{t_k}^{t_{k+1}} \left(\mathbb E \left[ \Big( (X_m(s) - X_m(t_k)) - (X(s)-X(t_k))\Big)^4\right]\right)^{\frac{1}{2}} \left(\mathbb E \left[ X^4_m(s) \right] \right)^{\frac{1}{2}}ds.
    \end{aligned}    
    \end{equation}
    \endgroup
    Using the same argument as in \eqref{proofeq: approximation of summands in NA-derivative 5}, one can verify that there exists a constant $C>0$ (not depending on $m$ or the partition) such that, for all $s\in[t_k, t_{k+1}]$,
    \begin{equation}\label{proofeq: approximation of summands in NA-derivative 7}
    \begin{aligned}
        \mathbb E \left[ (X(s) - X(t_k))^4 \right] \le C(s-t_k)^2.
    \end{aligned}    
    \end{equation}
    Moreover, using \eqref{eq: sup of S and X} and Theorem \ref{th: approximation of volatility}, it is easy to show that for any $s\in[0,T]$
    \begin{align*}
        \mathbb E \left[ X^4(s) (Y_m(s) - Y(s))^4 \right] & \le \left(\mathbb E \left[ X^8(s)\right]\right)^{\frac{1}{2}} \left(\mathbb E \left[ (Y_m(s) - Y(s))^8 \right]\right)^{\frac{1}{2}}  \le C \lVert \mathcal K - \mathcal K_m \rVert_{L^2([0,T])}^4,
    \end{align*}
    whence
    \begin{equation}\label{proofeq: approximation of summands in NA-derivative 8}
    \begin{aligned}
        \mathbb E &\left[ (Y_m(s) X_m(s) - Y(s)X(s))^4 \right] 
        \\
        & \le C \mathbb E \left[ Y^4_m(s) (X_m(s) - X(s))^4 \right] + C\mathbb E \left[ X^4(s) (Y_m(s) - Y(s))^4 \right]
        \\
        &\le C \mathbb E \left[ (X_m(s) - X(s))^4 \right] + C\mathbb E \left[ X^4(s) (Y_m(s) - Y(s))^4 \right] \le C\lVert \mathcal K - \mathcal K_m \rVert_{L^2([0,T])}^4.
    \end{aligned}    
    \end{equation}
    Next, by Lemma \ref{lemma: difference between increments}, for all $s\in[t_k, t_{k+1}]$,
    \begin{equation}\label{proofeq: approximation of summands in NA-derivative 9}
    \begin{aligned}
        \mathbb E \left[ \Big( (X_m(s) - X_m(t_k)) - (X(s)-X(t_k))\Big)^4 \right] &\le C(s-t_k)^2 \lVert \mathcal K - \mathcal K_m \rVert_{L^2([0,T])}^4.
    \end{aligned}    
    \end{equation}
    Plugging the estimates \eqref{proofeq: approximation of summands in NA-derivative 7}--\eqref{proofeq: approximation of summands in NA-derivative 9} into \eqref{proofeq: approximation of summands in NA-derivative 6}, we finally obtain that there exists a constant $C>0$ (not depending on $m$ or the partition) such that
    \begin{equation*}
    \begin{aligned}
        \mathbb E&\Bigg[\bigg( \int_{t_k}^{t_{k+1}} \Big((X_m(s) - X_m(t_k)) Y_m(s)X_m(s)- (X(s) - X(t_k)) Y(s)X(s)\Big) dW(s) \bigg)^2\Bigg]
        \\
        &\le C \int_{t_k}^{t_{k+1}} (s-t_k) \lVert \mathcal K - \mathcal K_m \rVert_{L^2([0,T])}^2 ds  + C \int_{t_k}^{t_{k+1}} (s-t_k) \lVert \mathcal K - \mathcal K_m \rVert_{L^2([0,T])}^2 \left(\sup_{m\ge 1} \mathbb E \left[ \sup_{t\in[0,T]}X^4_m(s) \right]\right)^{\frac{1}{2}}ds
        \\
        & \le C\Delta^2 \lVert \mathcal K - \mathcal K_m \rVert_{L^2([0,T])}^2.
    \end{aligned}    
    \end{equation*}  
    Taking into account all of the above, we can finally write
    \begin{equation*}
    \begin{aligned}
        \mathbb E \left[ |(\Delta X_m(t_k))^2  -  (\Delta X(t_k))^2|^2 \right] & \le C\Delta^2 \lVert \mathcal K - \mathcal K_m \rVert_{L^2([0,T])}^2
    \end{aligned}    
    \end{equation*}
    and whence we have the estimate
    \begin{equation*}
    \begin{aligned}
        I_3 & := \mathbb E \left[ |F_m| \left| \Delta X_m(t_k) \right|\frac{\mathbb E \left[ |(\Delta X_m(t_k))^2  -  (\Delta X(t_k))^2|~|~\mathcal F_{t_k} \right]}{ \mathbb E \left[ (\Delta X(t_k))^2~|~\mathcal F_{t_k} \right] \mathbb E \left[ (\Delta X_m(t_k))^2~|~\mathcal F_{t_k} \right]} \right]
        \\
        &\le \frac{C}{\Delta^2} \left(\mathbb E\left[ |\Delta X_m(t_k)|^4 \right]\right)^{\frac{1}{4}} \left(\mathbb E \left[ |(\Delta X_m(t_k))^2  -  (\Delta X(t_k))^2|^2 \right]\right)^{\frac{1}{2}}
        \\
        &\le \frac{C}{\Delta^2} \sqrt{\Delta} \Delta \lVert \mathcal K - \mathcal K_m \rVert_{L^2([0,T])} \le \frac{C}{\sqrt{\Delta}} \lVert \mathcal K - \mathcal K_m \rVert_{L^2([0,T])},
    \end{aligned}    
    \end{equation*}
    which finalizes the proof.
\end{proof}

\begin{theorem}\label{th: NA-derivative approximation}
    Let Assumptions \ref{assum: kernel}, \ref{assum: drift} and \ref{assum: approx kernels} hold, the payoff function $f = f_1 + f_2$ satisfy Assumption \ref{assum: F}, the kernels $\mathcal K$ and $\mathcal K_m$, $m\ge 1$, have the form
    \[
        \mathcal K(t,s) = \mathcal K(t-s)\mathbbm 1_{s<t}, \quad \mathcal K_m(t,s) = \mathcal K_m(t-s)\mathbbm 1_{s<t}, \quad t,s \in [0,T],
    \]
    and $\lVert \mathcal K - \mathcal K_m\rVert_{L^2([0,T])} \to 0$, $m\to \infty$. Then the following statements hold.
    \begin{itemize}
        \item[1)] There exists a constant $C>0$ that does not depend on $m$ such that, for any partition $\pi = \{0 = t_0 < t_1 < ... < t_n = T\}$ with the mesh $|\pi| := \max_{k}|t_{k+1} - t_k|$,
        \begin{equation}\label{eq: the most important formula 0}
            \mathbb E\left[ \int_0^T |u_{\pi}(s) - u^m_{\pi}(s)|ds \right] \le C n\sqrt{|\pi|}  \lVert \mathcal K - \mathcal K_m \rVert_{L^2([0,T])}^{\frac{1}{5}},
        \end{equation}
        where $u_\pi$ and $u^m_\pi$ are defined by \eqref{eq: definition of u} and \eqref{eq: definition of um}. In particular, if the partition $\pi$ is uniform, i.e. $t_k = \frac{kT}{n}$ and $|\pi| = \frac{T}{n}$,
        \begin{equation}\label{eq: the most important formula}
            \mathbb E\left[ \int_0^T |u_\pi(s) - u^m_\pi(s)|ds \right] \le C \sqrt{n} \lVert \mathcal K - \mathcal K_m \rVert_{L^2([0,T])}^{\frac{1}{5}}.
        \end{equation}
        
        \item[2)] If $\{m_n,~n\ge 1\}$ is such that
        \[
            n \sqrt{|\pi|} \lVert \mathcal K - \mathcal K_{m_n} \rVert_{L^2([0,T])}^{\frac{1}{5}} \to 0, \quad n \to \infty,
        \]
        then
        \[
            \mathbb E \left[ \int_0^T | \mathfrak D F (s) - u^{m_n}_{\pi}(s) |ds \right] \to 0, \quad n \to \infty.
        \]
    \end{itemize}
\end{theorem}

\begin{proof}
    By Theorem \ref{th: absolute continuity of laws of S and X}, the random variable $X(T)$ has continuous and bounded density, whence conditions of Theorem \ref{th: error in payoff} are fulfilled and
    \begin{align*}
       \mathbb E[|F - F_m|^4] & = \mathbb E[|f(X(T)) - f(X_m(T))|^4] \le C\left( \mathbb E[|X(T) - X_m(T)|^{4}] + \left(\mathbb E\left[ |X(T) - X_m(T)|^4\right]\right)^{\frac{4}{5}} \right)
       \\
       & \le C \left(\lVert \mathcal K - \mathcal K_m \rVert_{L^2([0,T])}^4 + \lVert \mathcal K - \mathcal K_m \rVert_{L^2([0,T])}^{\frac{4}{5}}\right) \le C\lVert \mathcal K - \mathcal K_m \rVert_{L^2([0,T])}^{\frac{4}{5}},
    \end{align*}
    where the constant $C$ does not depend on $m$ or the partition. Furthermore, by Lemma \ref{lemma: difference betweem u and um},
    \begin{align*}
        \mathbb E&\left[ \int_0^T |u_\pi(s) - u^m_\pi(s)|ds \right] = \sum_{k=0}^{n-1} \mathbb E\left[ |u_{\pi,k} - u^m_{\pi,k}| \right](t_{k+1} - t_k)
        \\
        &\le C \left( \sum_{k=0}^{n-1} \sqrt{t_{k+1} - t_k} \right)\left( \left(\mathbb E[|F - F_m|^4]\right)^{\frac{1}{4}} + \lVert \mathcal K - \mathcal K_m \rVert_{L^2([0,T])}\right)
        \\
        & \le C n\sqrt{|\pi|} \left( \lVert \mathcal K - \mathcal K_m \rVert_{L^2([0,T])}^{\frac{1}{5}} + \lVert \mathcal K - \mathcal K_m \rVert_{L^2([0,T])}\right)
        \\
        & \le C n\sqrt{|\pi|} \lVert \mathcal K - \mathcal K_m \rVert_{L^2([0,T])}^{\frac{1}{5}},
    \end{align*}
    which gives item 1).
    
    To prove item 2), first note that
    \begin{align*}
        \mathbb E &\left[ \int_0^T \left|\mathfrak D F(s) - u_\pi(s)\right| ds \right] = \mathbb E \left[ \int_0^T \left|\mathfrak D F(s) - u_\pi(s)\right| Y(s) X(s) \frac{1}{Y(s) X(s)} ds \right]
        \\
        & \le \left( \mathbb E \left[\int_0^T \left|\mathfrak D F(s) - u_\pi(s)\right|^2 Y^2(s) X^2(s) ds\right] \right)^{\frac{1}{2}} \left( \mathbb E \left[\int_0^T \frac{1}{Y^2(s) X^2(s)} ds \right]\right)^{\frac{1}{2}}
        \\
        & \le C \lVert \mathfrak D F - u_\pi \rVert_{L^2(\mathbb P \times [X])} \to 0
    \end{align*}
    as $|\pi| \to 0$. Whence
    \begin{align*}
         \mathbb E &\left[ \int_0^T | \mathfrak D F (s) - u^{m_n}_{\pi}(s) |ds \right] \le  \mathbb E \left[ \int_0^T \left| \mathfrak D F (s) - u_{\pi}(s) \right|ds \right] + \mathbb E \left[ \int_0^T \left| u_{\pi}(s) - u^{m_n}_{\pi}(s) \right|ds\right]
         \\
         & \le \mathbb E \left[ \int_0^T \left| \mathfrak D F (s) - u_{\pi}(s) \right|ds \right] + C n \sqrt{|\pi|} \lVert \mathcal K - \mathcal K_{m_n} \rVert_{L^2([0,T])}^{\frac{1}{5}} \to 0, \quad n\to \infty.
    \end{align*}
    By this, the proof is complete.
\end{proof}

We remark that that the exponent ${\frac{1}{5}}$ appears in \eqref{eq: the most important formula 0} and \eqref{eq: the most important formula} exclusively due the estimate \eqref{eq: error in payoff discontinuous} of Theorem \ref{th: error in payoff} that corresponds to the (possibly discontinuous) component $f_2$ of $f$. If $f_2 \equiv 0$, i.e. when
\begin{align*}
    \mathbb E[|F - F_m|^4] &= \mathbb E[|f(X(T)) - f(X_m(T))|^4] \le C \mathbb E [|X(T) - X_m(T)|^4] 
    \\
    &\le C \lVert \mathcal K - \mathcal K_m\rVert_{L^2([0,T])}^4,
\end{align*}
Theorem \ref{th: NA-derivative approximation} can be reformulated as follows.

\begin{theorem}\label{th: NA-derivative approximation Lipschitz}
    Let Assumptions \ref{assum: kernel}, \ref{assum: drift} and \ref{assum: approx kernels} hold, the payoff function $f = f_1$ be globally Lipschitz, $\mathcal K$ and $\mathcal K_m$, $m\ge 1$, have the form
    \[
        \mathcal K(t,s) = \mathcal K(t-s)\mathbbm 1_{s<t}, \quad \mathcal K_m(t,s) = \mathcal K_m(t-s)\mathbbm 1_{s<t}, \quad t,s \in [0,T],
    \]
    and $\lVert \mathcal K - \mathcal K_m\rVert_{L^2([0,T])} \to 0$, $m\to \infty$. Then the following statements are true.
    \begin{itemize}
        \item[1)] There exists a constant $C>0$ that does not depend on $m$ such that for any partition $\pi = \{0 = t_0 < t_1 < ... < t_n = T\}$ with the mesh $|\pi| := \max_{k}|t_{k+1} - t_k|$
        \[
            \mathbb E\left[ \int_0^T |u_{\pi}(s) - u^m_{\pi}(s)|ds \right] \le C n\sqrt{|\pi|}  \lVert \mathcal K - \mathcal K_m \rVert_{L^2([0,T])},
        \]
        where $u_\pi$ and $u^m_\pi$ are defined by \eqref{eq: definition of u} and \eqref{eq: definition of um}. In particular, if $\pi$ is uniform, i.e. $t_k = \frac{kT}{n}$ and $|\pi| = \frac{T}{n}$,
        \begin{equation}\label{eq: the most important formula Lipschitz}
            \mathbb E\left[ \int_0^T |u_\pi(s) - u^m_\pi(s)|ds \right] \le C \sqrt{n} \lVert \mathcal K - \mathcal K_m \rVert_{L^2([0,T])}.
        \end{equation}
        
        \item[2)] If $\{m_n,~n\ge 1\}$ is such that
        \[
            n \sqrt{|\pi|} \lVert \mathcal K - \mathcal K_{m_n} \rVert_{L^2([0,T])} \to 0, \quad n \to \infty,
        \]
        then
        \[
            \mathbb E \left[ \int_0^T | \mathfrak D F (s) - u^{m_n}_{\pi}(s) |ds \right] \to 0, \quad n \to \infty.
        \]
    \end{itemize}
\end{theorem}

\begin{example}\label{ex: NA-derivative Holder}
    Let the kernel $\mathcal K$  be as in Examples \ref{ex: Holder kernel} and \ref{ex: Holder kernel approx}, i.e. $\mathcal K(t,s) = \mathcal K(t-s) \mathbbm 1_{s<t}$, $\mathcal K(0) = 0$ and $\mathcal K \in C^H([0,T])$ for some $H\in (0,1)$. Let $X$ be the corresponding SVV discounted price \eqref{eq: X},
    $\mathcal K_m$ be the Bernstein polynomial approximation of $\mathcal K$ given by \eqref{eq: K Bernstein} and $X_m$ be the approximation \eqref{eq: X approximated} of $X$ constructed using $\mathcal K_m$. Then, taking into account \eqref{eq: Bernstein kernel error} and Theorem \ref{th: approx with difference kernel}, for any $r\ge 2$
    \[
        \mathbb E \left[ \sup_{t\in [0,T]}|X(t) - X_m(t)|^r \right] \le C m^{-\frac{rH}{2}} \to 0, \quad m\to \infty.
    \]
    Moreover, for a uniform partition $\pi = \{0 = t_0 < t_1 < ... < t_n = T\}$, $t_k = \frac{kT}{n}$, Theorem \ref{th: NA-derivative approximation} and Theorem \ref{th: NA-derivative approximation Lipschitz} imply that
    \begin{itemize}
        \item[1)] if $f = f_1 + f_2$ satisfies Assumption \ref{assum: F} and $m_n := n^{\alpha}$ for $\alpha > \frac{5}{H}$, then
        \begin{equation}\label{eq: NA-derivative for Holder kernels}
            \mathbb E\left[ \int_0^T |u_\pi(s) - u^{m_n}_\pi(s)|ds \right] \le C n^{-\frac{\alpha H}{10} + \frac{1}{2}} \to 0, \quad n \to \infty;
        \end{equation}
        
        \item[2)] if $f = f_1$ is globally Lipschitz and $m_n := n^{\alpha}$ for $\alpha > \frac{1}{H}$, then
        \begin{equation}\label{eq: NA-derivative for Holder kernels Lipschitz}
            \mathbb E\left[ \int_0^T |u_\pi(s) - u^{m_n}_\pi(s)|ds \right] \le C n^{-\frac{\alpha H}{2} + \frac{1}{2}} \to 0, \quad n \to \infty.
        \end{equation}
    \end{itemize}
\end{example}

\begin{example}\label{ex: NA-derivative fractional}
    Let the kernel $\mathcal K$  be the fractional kernel from Examples \ref{ex: fractional kernel} and \ref{ex: fractional kernel approx}, $X$ be the corresponding SVV discounted price \eqref{eq: X}, $\mathcal K_m$ be the approximation of $\mathcal K$ with exponentials given by \eqref{eq: rough fractional kernel approximation} with $\tau_{m,i}$, $i=0,...,m$, $m\ge 1$, given by \eqref{eq: possible choice of tau}. Let also $X_m$ be the approximation \eqref{eq: X approximated} of $X$ constructed using $\mathcal K_m$. Then, taking into account \eqref{eq: fractional kernel error} and Theorem \ref{th: approx with difference kernel}, for any $r\ge 2$,
    \[
        \mathbb E \left[ \sup_{t\in[0,T]}|X(t) - X_m(t)|^r \right] \le C m^{-\frac{4rH}{5}}.
    \]
    Moreover, for a uniform partition $\pi = \{0 = t_0 < t_1 < ... < t_n = T\}$, $t_k = \frac{kT}{n}$, Theorem \ref{th: NA-derivative approximation} and Theorem \ref{th: NA-derivative approximation Lipschitz} imply that
    \begin{itemize}
        \item[1)] if $f = f_1 + f_2$ satisfies Assumption \ref{assum: F} and $m_n := n^{\alpha}$ for $\alpha > \frac{25}{8H}$, then
        \begin{equation}\label{eq: NA-derivative for rough kernels}
            \mathbb E\left[ \int_0^T |u_\pi(s) - u^{m_n}_\pi(s)|ds \right] \le C n^{-\frac{4\alpha H}{25} + \frac{1}{2}} \to 0, \quad n \to \infty;
        \end{equation}
        
        \item[2)] if $f = f_1$ is globally Lipschitz and $m_n := n^{\alpha}$ for $\alpha > \frac{5}{8H}$, then
        \begin{equation}\label{eq: NA-derivative for rough kernels Lipschitz}
            \mathbb E\left[ \int_0^T |u_\pi(s) - u^{m_n}_\pi(s)|ds \right] \le C n^{-\frac{4\alpha H}{5} + \frac{1}{2}} \to 0, \quad n \to \infty.
        \end{equation}
    \end{itemize}
\end{example}

\section{Monte Carlo computation of the optimal hedge and simulations}\label{s: Monte Carlo approximation of the hedging strategy}

In Section \ref{s: NA-derivative estimates}, we approximated the optimal hedging portfolio $\mathfrak D f(X(T))$ with $\mathfrak D f(X_m(T))$, where $X_m$ is a coordinate of $(m+2)$-dimensional Markov process. As noted above, Markovianity is incredibly beneficial for the numerical computation of conditional expectations in \eqref{eq: definition of um}: indeed,
\begin{align*}
    \mathbb E &\left[f(X_m(T))(X_m(t_{k+1}) - X_m(t_k))~|~\mathcal F_{t_k}\right]  
    \\
    &= \mathbb E \left[f(X_m(T))(X_m(t_{k+1}) - X_m(t_k))~|~X_m(t_k),~Y_m(t_k),~U_{m,0}(t_k), ..., U_{m,m}(t_k)\right]
    \\
    &=: \Phi_1\big(t_k, X_m(t_k),Y_m(t_k),U_{m,0}(t_k), ..., U_{m,m}(t_k)\big),
    \\
    \mathbb E&\left[ (X_m(t_{k+1}) - X_m(t_k))^2~|~\mathcal F_{t_k}\right]  
    \\
    &= \mathbb E\left[ (X_m(t_{k+1}) - X_m(t_k))^2~|~X_m(t_k),~Y_m(t_k),~U_{m,0}(t_k), ..., U_{m,m}(t_k)\right]
    \\
    &=: \Phi_2\big(t_k,  X_m(t_k),Y_m(t_k),U_{m,0}(t_k), ..., U_{m,m}(t_k)\big),
\end{align*}
i.e. the computational problem boils down to learning the shape of functions $\Phi_1$ and $\Phi_2$ that can be done via Monte Carlo methods.

In this Section, we propose two algorithms for estimating $\Phi_1$ and $\Phi_2$: Nested Monte Carlo (NMC) and Least-Squares Monte Carlo (LSMC). In order to simulate paths of sandwiched volatility processes \eqref{eq: Y}, we use the drift-implicit Euler approximation scheme from \cite{DNMYT2022-1}. The original and the approximated discounted price processes $X$ and $X_m$ are simulated just like in \cite{DNMYT2022}. We also present simulations of the hedging strategies for a standard European call option. Note that algorithms presented below also work for exotic contracts with discontinuous payoffs, but slower convergence rates require longer computations. All simulations were performed in \textsf{R} programming language on the system with Intel Core i9-9900K CPU and 64 Gb RAM. 

\subsection{Nested Monte Carlo method}

The first and more straightforward way to compute $\Phi_1$ and $\Phi_2$ is the \textit{Nested Monte Carlo} (NMC) approach that can be summarized as follows (see also Fig. \ref{fig: NMC}):
\begin{itemize}
    \item[1)] given $X_m(t_k) = x$, $Y_m(t_k) = y$, $U_{m,0}(t_k) = u_0$, ..., $U_{m,m}(t_k) = u_m$,
    simulate $N$ independent trajectories $\{X^{(i)}_m(t),~t \in (t_k,T]\}$, $i=1,...,N$;
    \item[2)] for each trajectory, compute $f(X^{(i)}_m(T))(X^{(i)}_m(t_{k+1}) - X^{(i)}_m(t_k))$ and $(X^{(i)}_m(t_{k+1}) - X^{(i)}_m(t_k))^2$, $i=1,...,N$;
    \item[3)] put 
    \begin{align*}
        \widehat \Phi_1(t_k, x, y, u_0, ..., u_m) &:= \frac{1}{N} \sum_{i=1}^N f(X^{(i)}_m(T))(X^{(i)}_m(t_{k+1}) - X^{(i)}_m(t_k)),
        \\
        \widehat \Phi_2(t_k, x, y, u_0, ..., u_m) & := \frac{1}{N} \sum_{i=1}^N (X^{(i)}_m(t_{k+1}) - X^{(i)}_m(t_k))^2.
    \end{align*}
\end{itemize}
\begin{figure}[h!]
    \centering
    \includegraphics[width = 0.6\textwidth]{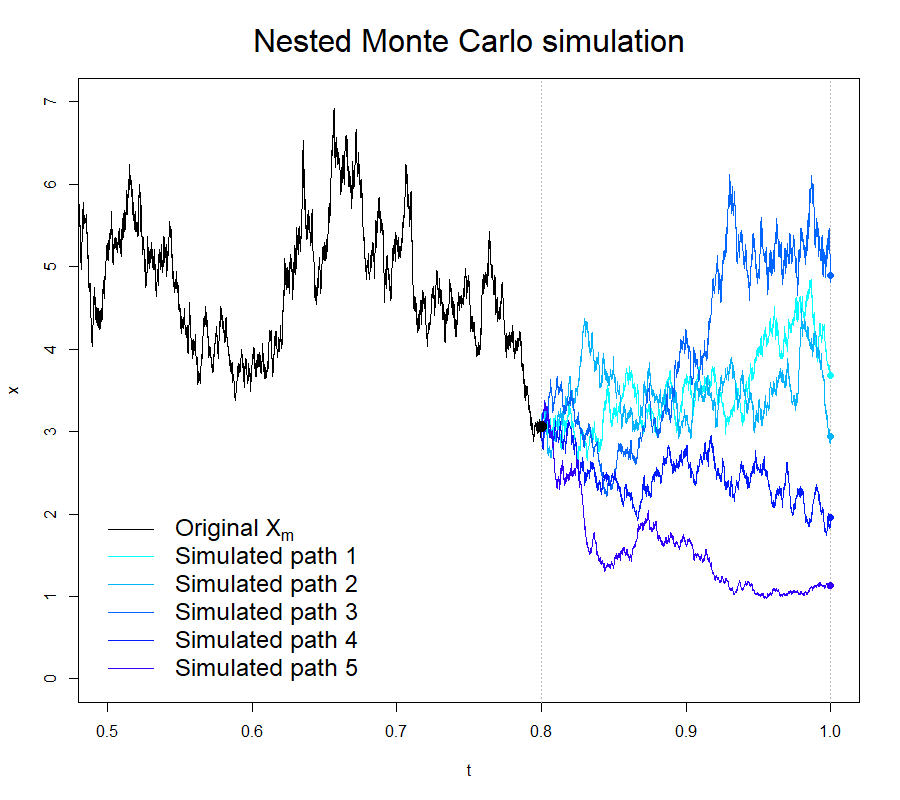}
    \caption{Nested Monte Carlo approach, $t_k = 0.8$, $T=1$. Given the values $X_m(t_k)$, $Y_m(t_k)$, $U_{m,1}(t_k)$, ..., $U_{m,m}(t_k)$, we simulate a number of trajectories $X^{(i)}_m$ on $(t_k, T]$ (blue), use each of those to compute $f(X^{(i)}_m(T))(X^{(i)}_m(t_{k+1}) - X^{(i)}_m(t_k))$ and $(X^{(i)}_m(t_{k+1}) - X^{(i)}_m(t_k))^2$. The means of the latter values are then used as approximations of the required conditional expectations.  }
    \label{fig: NMC}
\end{figure}

\begin{example}\label{ex: NMC Holder}{\textit{(H\"older continuous kernel)}}
    Consider the SVV model
    \begin{align*}
        X(t) &= X(0) + \int_0^t Y(s) X(s) \left(\rho dB_1(s) + \sqrt{1 - \rho^2} dB_2(s)\right), 
        \\
        Y(t) &= Y(0) + \int_0^t b(s, Y(s))ds + \int_0^t \mathcal K(t,s)dB_1(s),
    \end{align*}
    with $X(0) = 5$, $Y(0) = 1$, $\rho = 0.5$, $b(t,y) = \frac{1}{(y - 0.01)^4} - \frac{1}{(5 - y)^4}$ and $\mathcal K(t,s) = \mathcal K(t-s) = (t-s)^{0.4}$. The approximation $(X_m, Y_m, U_{m,0}, U_{m,1}, ..., U_{m,m})$ is constructed using the Bernstein polynomial approximation $\mathcal K_m$ from \eqref{eq: K Bernstein}. 
    
    \begin{figure}[h!]
    \centering
    \begin{minipage}[b]{0.4\linewidth}
        \centering
        \includegraphics[width=\textwidth]{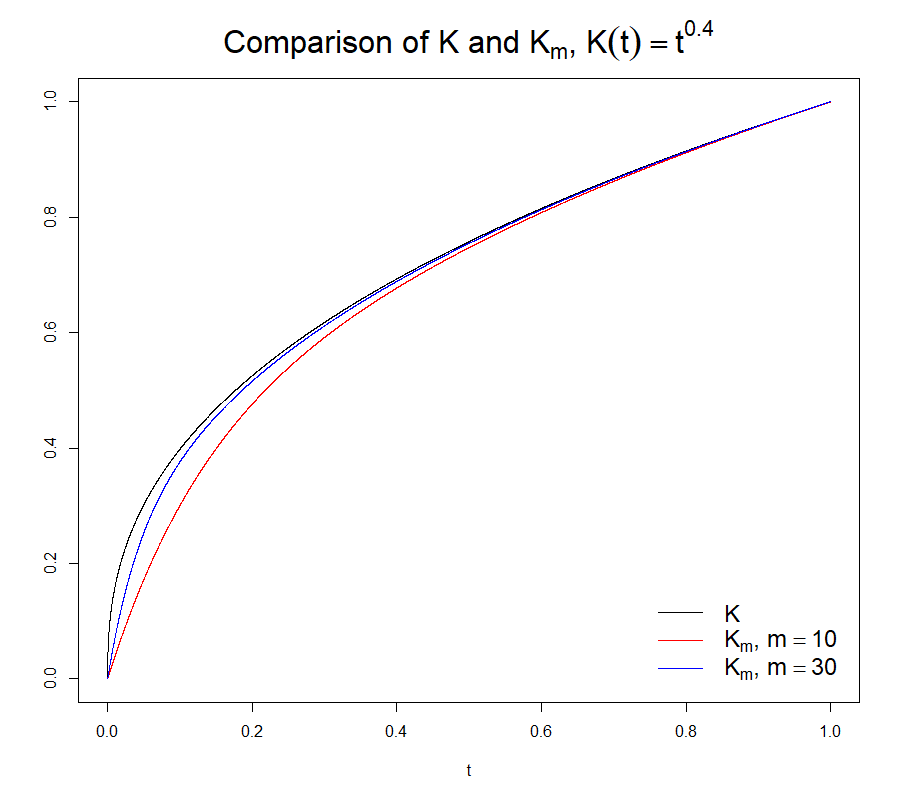}
        (a) Kernels $\mathcal K$ and $\mathcal K_m$
    \end{minipage}
    \begin{minipage}[b]{0.4\linewidth}
        \centering
        \includegraphics[width=\textwidth]{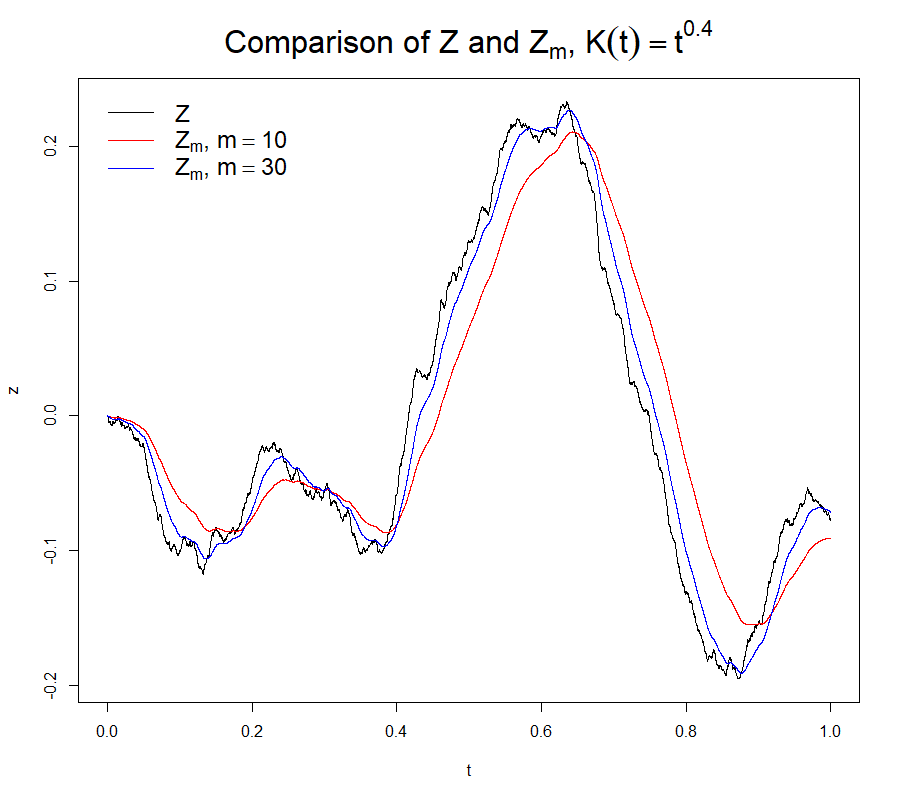}
        (b) Volterra noises $Z$ and $Z_m$
    \end{minipage}
    \begin{minipage}[b]{0.4\linewidth}
        \centering
        \includegraphics[width=\textwidth]{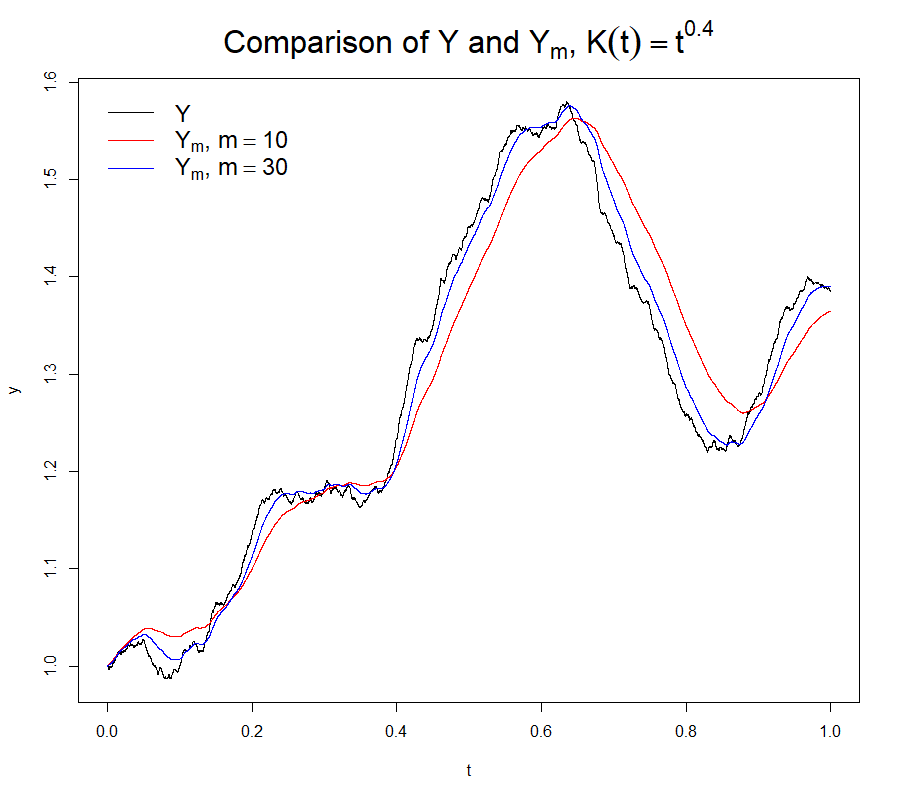}
        (c) Volatility processes $Y$ and $Y_m$
    \end{minipage}
    \begin{minipage}[b]{0.4\linewidth}
        \centering
        \includegraphics[width=\textwidth]{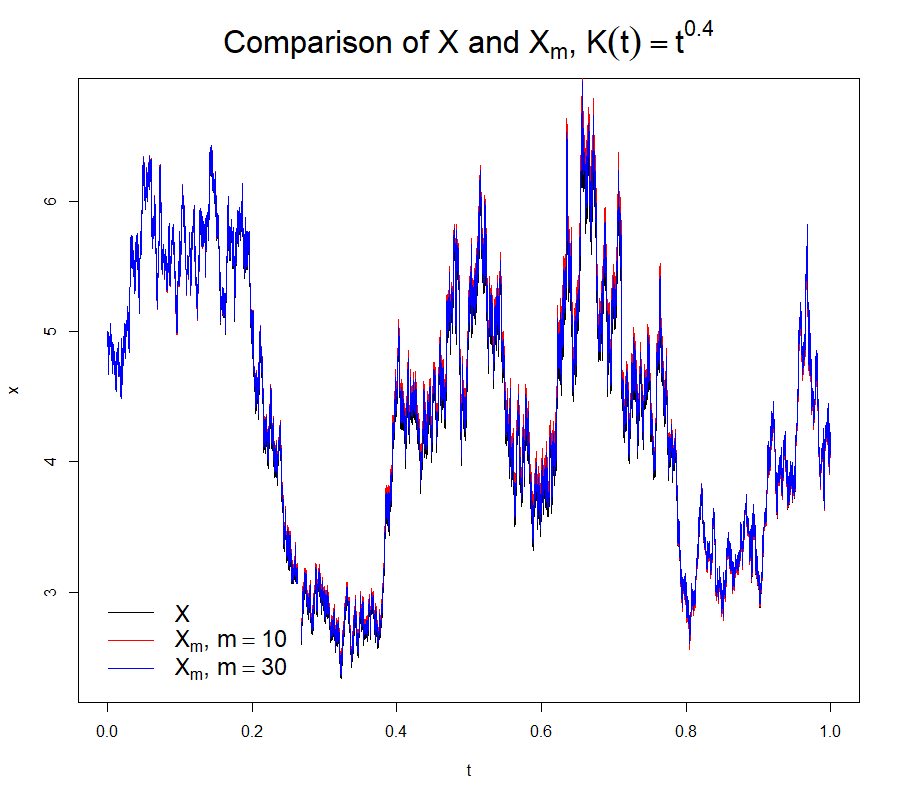}
        (d) Price processes $X$ and $X_m$
    \end{minipage}
    \caption{Approximation of the SVV model with H\"older continuous kernel}
    \label{fig: NMC Holder}
    \end{figure}
    \noindent Figure \ref{fig: NMC Holder} illustrates the approximation of the SVV model; original kernel and generated sample paths of the corresponding processes are depicted in black whereas approximations using the Bernstein polynomials are depicted in red (of order $m=10$) and blue (of order $m=30$). Note that the path of $X$ simulated using the original noise (black) is not visible on Figure \ref{fig: NMC Holder}(d) due to the high degree of overlapping with the approximations (red and blue trajectories). Figure \ref{fig: NMC hedge} contains a path of the process
    \begin{align}
        \widehat u_{\pi}(t) &:= \sum_{k=0}^{9} \frac{\widehat\Phi_1\left(\frac{k}{10}, X_m\left(\frac{k}{10}\right), Y_m\left(\frac{k}{10}\right), U_{m,0}\left(\frac{k}{10}\right), ..., U_{m,m}\left(\frac{k}{10}\right)\right)}{\widehat\Phi_2\left(\frac{k}{10}, X_m\left(\frac{k}{10}\right), Y_m\left(\frac{k}{10}\right), U_{m,0}\left(\frac{k}{10}\right), ..., U_{m,m}\left(\frac{k}{10}\right)\right)} \mathbbm 1_{\left(\frac{k}{10}, \frac{k+1}{10}\right]}(t),
    \end{align}
    with the European call payoff function $f(x) := \max\{x-4, 0\}$ and the values $X_m\left(\frac{k}{10}\right)$, $Y_m\left(\frac{k}{10}\right)$, $U_{m,0}\left(\frac{k}{10}\right)$, ..., $U_{m,m}\left(\frac{k}{10}\right)$ coming exactly from the trajectory depicted on Figure \ref{fig: NMC Holder}. In order to estimate
    \begin{gather*}
        \widehat\Phi_1\left(\frac{k}{10}, X_m\left(\frac{k}{10}\right), Y_m\left(\frac{k}{10}\right), U_{m,0}\left(\frac{k}{10}\right), ..., U_{m,m}\left(\frac{k}{10}\right)\right),
        \\
        \widehat\Phi_2\left(\frac{k}{10}, X_m\left(\frac{k}{10}\right), Y_m\left(\frac{k}{10}\right), U_{m,0}\left(\frac{k}{10}\right), ..., U_{m,m}\left(\frac{k}{10}\right)\right),
    \end{gather*}
    100000 simulations were used for each $k=0,1,...,9$.
    \begin{figure}[h!]
        \centering
        \includegraphics[width = 0.5\textwidth]{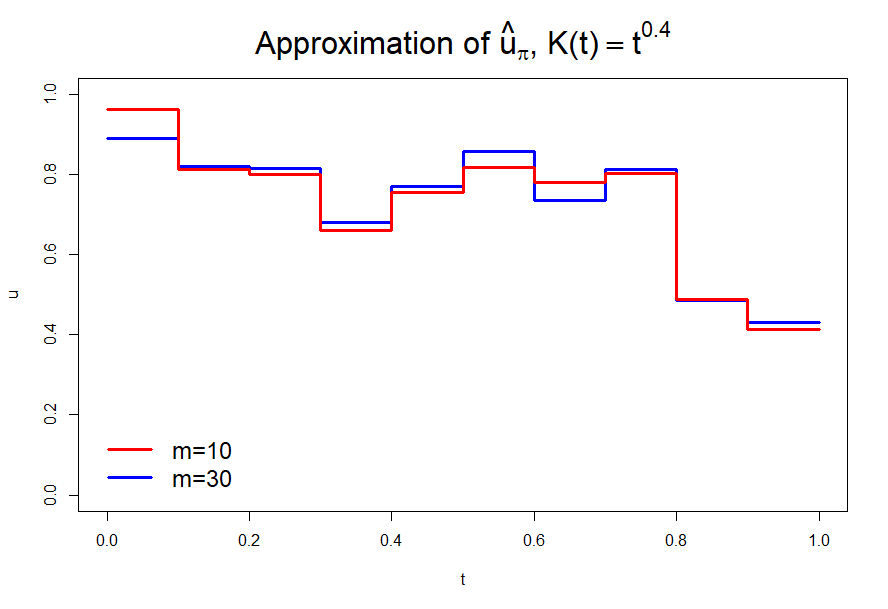}
        \caption{Hedging strategy estimated for the path from Fig. \ref{fig: NMC Holder} for $m=10$ (red) and $m=30$ (blue). The corresponding partition is $k/10$, $k=0,1,...,10$. Simulating the red line took 63231 seconds whereas the blue line -- 73410 seconds.}
    \label{fig: NMC hedge}
\end{figure}

\end{example}

\begin{example}\label{ex: NMC rough}{\textit{(Rough kernel)}}
    Consider the SVV model
    \begin{align*}
        X(t) &= X(0) + \int_0^t Y(s) X(s) \left(\rho dB_1(s) + \sqrt{1 - \rho^2} dB_2(s)\right),
        \\
        Y(t) &= Y(0) + \int_0^t b(s, Y(s))ds + \int_0^t \mathcal K(t,s)dB_1(s),
    \end{align*}
    with $X(0) = 5$, $Y(0) = 1$, $\rho = 0.5$, $b(t,y) = \frac{1}{(y - 0.01)^4} - \frac{1}{(5 - y)^4}$ and $\mathcal K(t,s) = \mathcal K(t-s) = \frac{1}{\Gamma(0.8) (t-s)^{0.2}}$. The approximation $(X_m, Y_m, U_{m,1}, ..., U_{m,m})$ is constructed by approximating the Volterra noise $Z(t) := \frac{1}{\Gamma(0.8)} \int_0^t \frac{1}{(t-s)^{0.2}} dB_1(s)$ by a linear combination of standard Ornstein-Uhlenbeck processes as described in Example \ref{ex: fractional kernel approx} and Remark \ref{rem: OU}. Just as in Example \ref{ex: NMC Holder}, Figure \ref{fig: NMC rough} illustrates approximation of the SVV model; original kernel and generated sample paths of the corresponding processes are depicted in black whereas approximations are depicted in red ($m=10$ summands), green ($m=100$ summands) and blue ($m=1000$ summands). Note that rough volatility requires much higher values of $m$ in order to ensure a decent level of approximation in comparison to Example \ref{ex: NMC Holder}. Figure \ref{fig: NMC rough hedge} contains a path of the process \begin{align}
        \widehat u_{\pi}(t) &:= \sum_{k=0}^{9} \frac{\widehat\Phi_1\left(\frac{k}{10}, X_m\left(\frac{k}{10}\right), Y_m\left(\frac{k}{10}\right), U_{m,0}\left(\frac{k}{10}\right), ..., U_{m,m}\left(\frac{k}{10}\right)\right)}{\widehat\Phi_2\left(\frac{k}{10}, X_m\left(\frac{k}{10}\right), Y_m\left(\frac{k}{10}\right), U_{m,0}\left(\frac{k}{10}\right), ..., U_{m,m}\left(\frac{k}{10}\right)\right)} \mathbbm 1_{\left(\frac{k}{10}, \frac{k+1}{10}\right]}(t),
    \end{align}
    $t\in[0,1]$, where the payoff function $f(x) := \max\{x-4, 0\}$ and the values $X_m\left(\frac{k}{10}\right)$, $Y_m\left(\frac{k}{10}\right)$, $U_{m,1}\left(\frac{k}{10}\right)$, ..., $U_{m,m}\left(\frac{k}{10}\right)$ are exactly the ones from the trajectory depicted on Figure \ref{fig: NMC rough}. In order to estimate
    \begin{gather*}
        \widehat\Phi_1\left(\frac{k}{10}, X_m\left(\frac{k}{10}\right), Y_m\left(\frac{k}{10}\right), U_{m,0}\left(\frac{k}{10}\right), ..., U_{m,m}\left(\frac{k}{10}\right)\right),
        \\
        \widehat\Phi_2\left(\frac{k}{10}, X_m\left(\frac{k}{10}\right), Y_m\left(\frac{k}{10}\right), U_{m,0}\left(\frac{k}{10}\right), ..., U_{m,m}\left(\frac{k}{10}\right)\right),
    \end{gather*}
    100000 simulations were used for each $k=0,1,...,9$.
    
    \begin{figure}[h!]
    \centering
    \begin{minipage}[b]{0.4\linewidth}
        \centering
        \includegraphics[width=\textwidth]{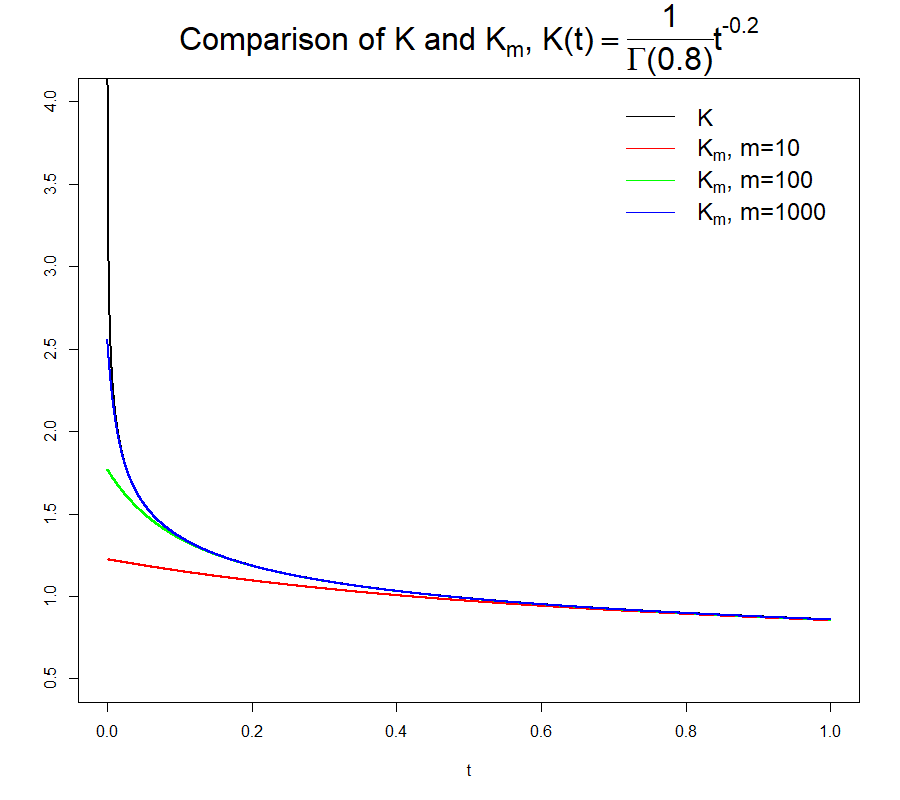}
        (a) Kernels $\mathcal K$ and $\mathcal K_m$
    \end{minipage}
    \begin{minipage}[b]{0.4\linewidth}
        \centering
        \includegraphics[width=\textwidth]{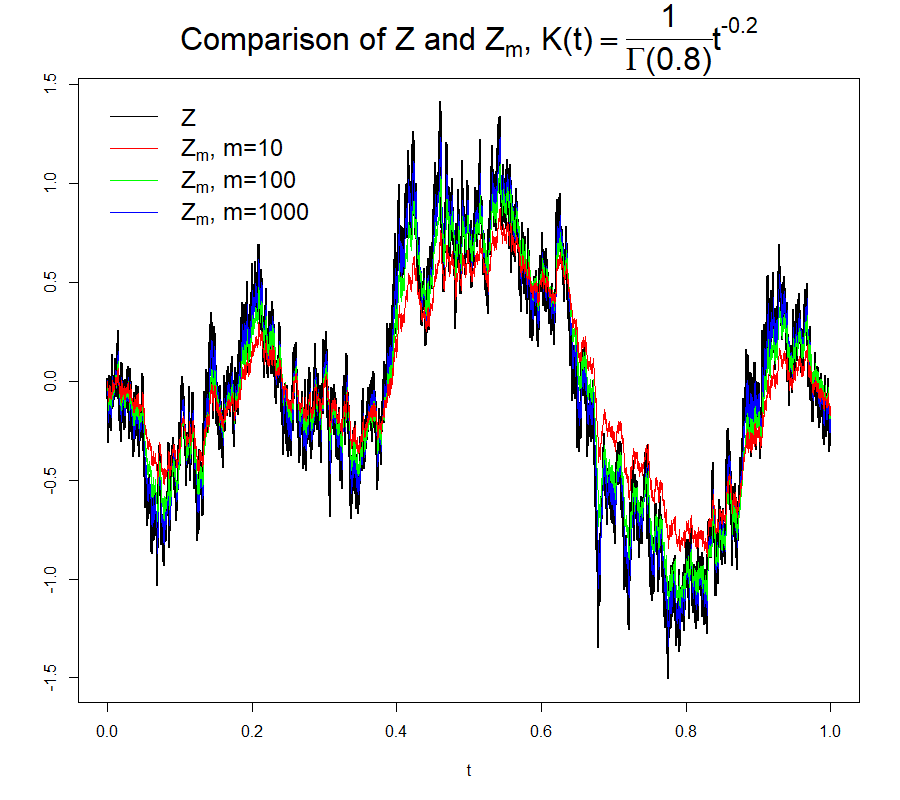}
        (b) Volterra noises $Z$ and $Z_m$
    \end{minipage}
    \begin{minipage}[b]{0.4\linewidth}
        \centering
        \includegraphics[width=\textwidth]{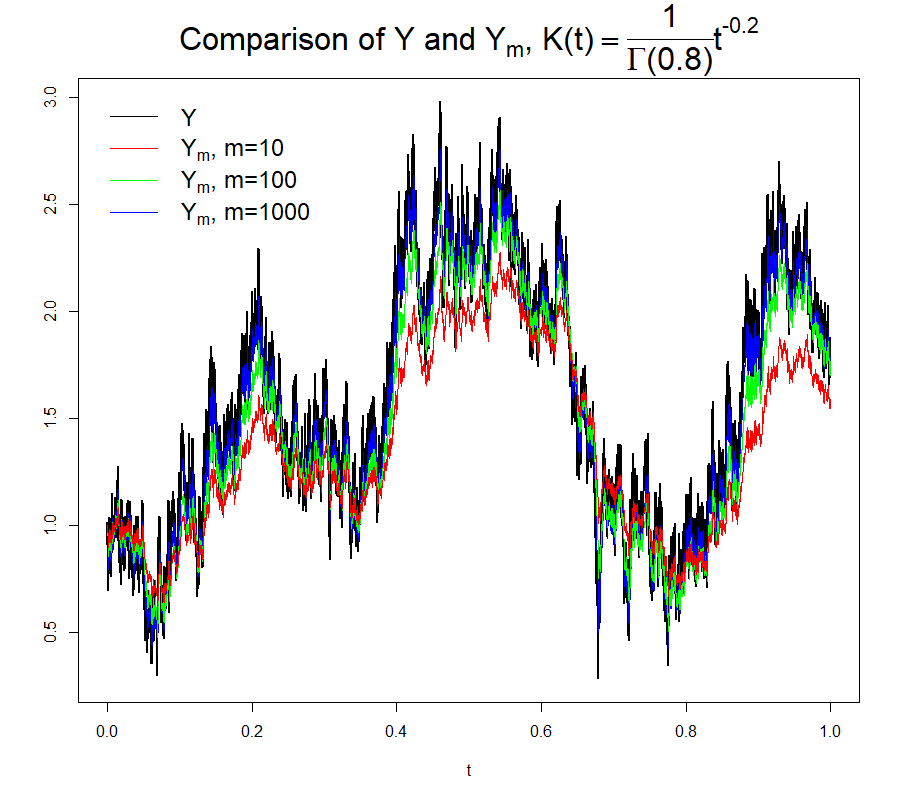}
        (c) Volatility processes $Y$ and $Y_m$
    \end{minipage}
    \begin{minipage}[b]{0.4\linewidth}
        \centering
        \includegraphics[width=\textwidth]{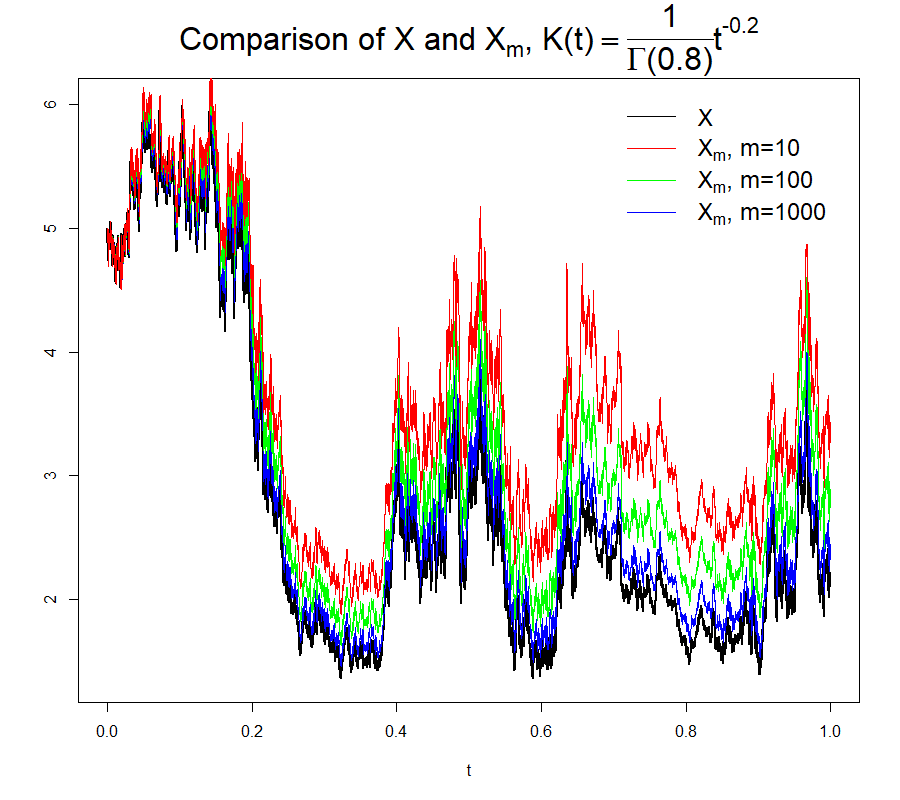}
        (d) Price processes $X$ and $X_m$
    \end{minipage}
    \caption{Approximation of the SVV model with rough fractional kernel}
    \label{fig: NMC rough}
    \end{figure}

    \begin{figure}[h!]
        \centering
        \includegraphics[width = 0.4\textwidth]{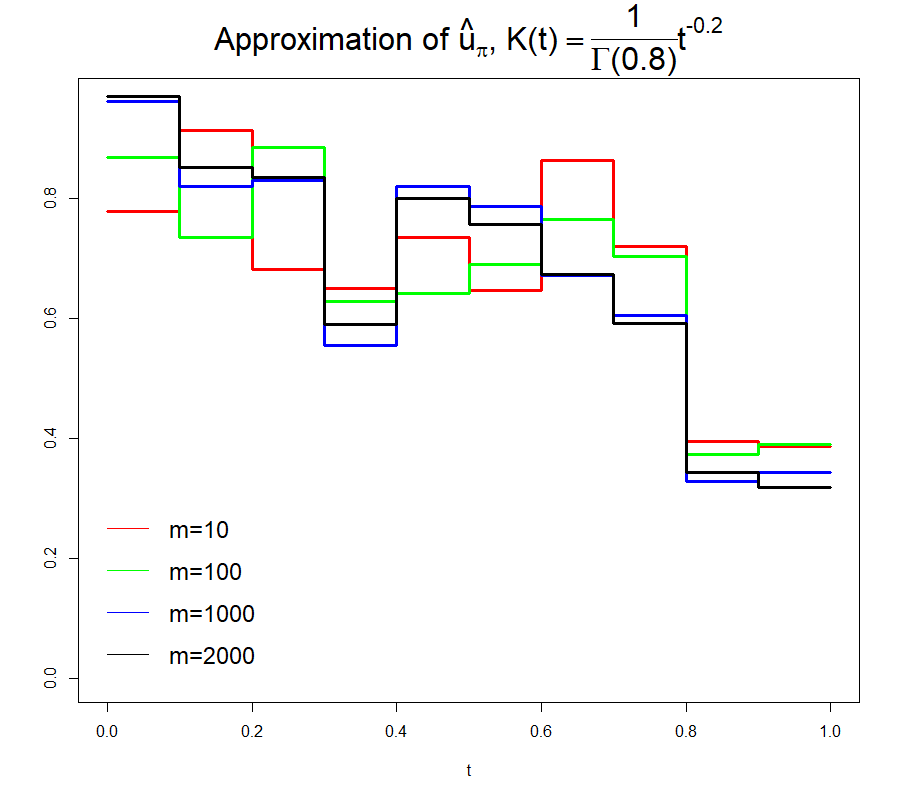}
        \caption{Hedging strategy estimated for the path from Fig. \ref{fig: NMC rough} for $m=10$ (red), $m=100$ (green), $m=1000$ (blue) and $m=2000$ (black). The corresponding partition is $k/10$, $k=0,1,...,10$. The figure also illustrates slower rate of convergence in comparison to the H\"older kernel case: the black and blue lines are close to each other but the red and green lines (corresponding to relatively low values of $m$) differ substantially. Computation time: 24162 seconds for $m=10$, 25348 seconds for $m=100$, 38530 seconds for $m=1000$ and 42431 seconds for $m=2000$. }
    \label{fig: NMC rough hedge}
    \end{figure}
\end{example}

\newpage

\subsection{Least squares Monte Carlo method}

Despite its simplicity and clear theoretical justification, the nested Monte Carlo approach has a substantial disadvantage: it takes long to compute and thus requires powerful computational resources in order to be used in practice. In order to overcome this issue, one can use the \textit{Least Squares Monte Carlo} (LSMC) method instead: 
\begin{itemize}
    \item[1)] simulate $N$ independent realizations
    \[
        (X^{(i)}_m,Y^{(i)}_m,U^{(i)}_{m,0},...,U^{(i)}_{m,m}) = \left\{\big(X^{(i)}_m(t),Y^{(i)}_m(t),U^{(i)}_{m,0}(t),...,U^{(i)}_{m,m}(t)\big),~t\in[0,T]\right\},
    \]
    $i = 1,...,N$;
    \item[2)] for each path $i=1,...,N$, evaluate 
    \[
        \big(X^{(i)}_m(t_k),Y^{(i)}_m(t_k),U^{(i)}_{m,0}(t_k),...,U^{(i)}_{m,m}(t_k)\big), \quad f(X^{(i)}_m(T))(X^{(i)}_m(t_{k+1}) - X^{(i)}_m(t_k)), \quad
        (X^{(i)}_m(t_{k+1}) - X^{(i)}_m(t_k))^2;
    \]
    \item[3)] apply an appropriate non-parametric regression (with a mean squared error loss function) to the generated ``dataset'' treating
    \[
        \big(X^{(i)}_m(t_k),Y^{(i)}_m(t_k),U^{(i)}_{m,0}(t_k),...,U^{(i)}_{m,m}(t_k)\big), \quad i = 1,...,N,
    \]
    as ``input'' and 
    \[
        f(X^{(i)}_m(T))(X^{(i)}_m(t_{k+1}) - X^{(i)}_m(t_k)), \quad (X^{(i)}_m(t_{k+1}) - X^{(i)}_m(t_k))^2, \quad i = 1,...,N,
    \]
    as ``output'' variables. The resulting estimates $\widehat \Phi_1(t_k, \cdot)$, $\widehat \Phi_2(t_k, \cdot)$ of the regression functions are then used to calculate $\widehat u_\pi$ on the interval $(t_k, t_{k+1}]$.
\end{itemize}

\begin{remark}
    The Least Squares Monte Carlo method described above can be regarded as a simplified version of the approach suggested in \cite{Carriere_1996, Longstaff_Schwartz_2001, Tsitsiklis_Van_Roy_2001} for pricing American options or in \cite{Krah_Nikolic_Korn_2018, Krah_Nikolic_Korn_2020b, Krah_Nikolic_Korn_2020} for modeling of life insurance companies.
\end{remark}

\begin{example}\label{ex: LSMC}
    Fig. \ref{fig: LSMC} contains approximations of the optimal hedging strategy constructed for the path of $X_m$ from Example \ref{ex: NMC Holder} that corresponds to $m=10$. In order to obtain the dark green line, we simulated the ``dataset'' containing the ``input'' variables
    \[
        \big(X^{(i)}_m(t_k),Y^{(i)}_m(t_k),U^{(i)}_{m,0}(t_k),...,U^{(i)}_{m,m}(t_k)\big), \quad i = 1,...,N,
    \]
    together with the corresponding ``output'' variables 
    \[
        f(X^{(i)}_m(T))(X^{(i)}_m(t_{k+1}) - X^{(i)}_m(t_k)), \quad (X^{(i)}_m(t_{k+1}) - X^{(i)}_m(t_k))^2, \quad i = 1,...,N,
    \]
    with $N=1000000$ for each $k=0,1,...,9$. In order to compute $\widehat \Phi_1(t_k, \cdot)$, $\widehat \Phi_2(t_k, \cdot)$, we used the idea from \cite{Krah_Nikolic_Korn_2020b} and utilized a neural network with 3 hidden layers and 20 nods in each layer. The red line is exactly the path of $\hat u_\pi$ for $m=10$ from Fig. \ref{fig: NMC hedge} and is treated as a reference. 
    \begin{figure}[h!]
        \centering
        \includegraphics[width = 0.5\textwidth]{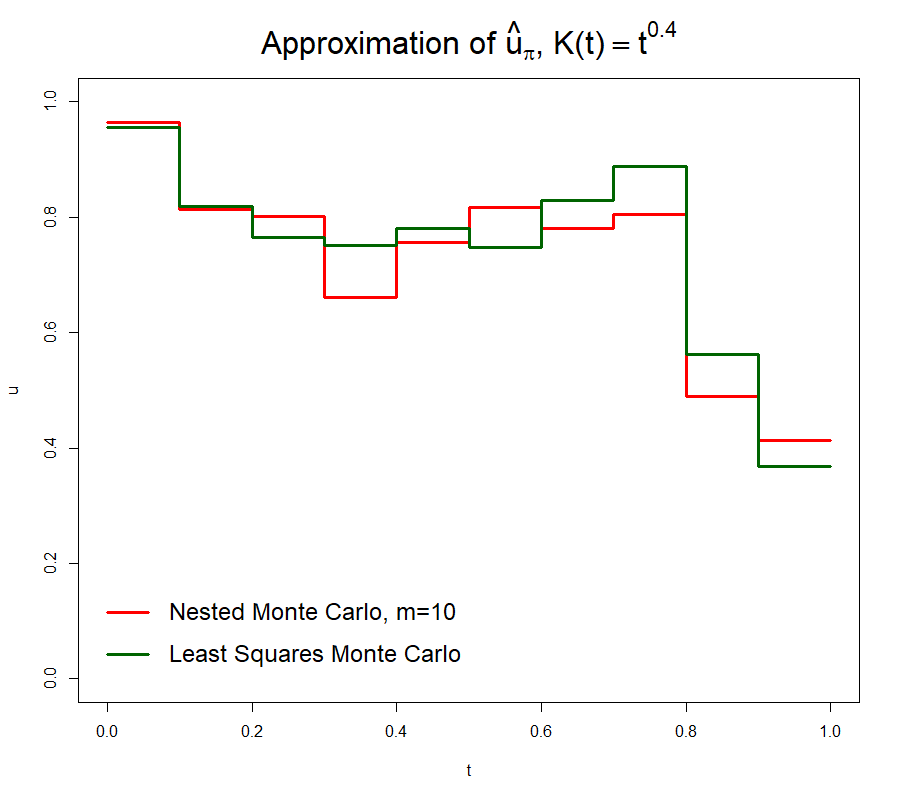}
        \caption{Approximation of the optimal hedge using the Nested Monte Carlo method for $m=10$ from Fig. \ref{fig: NMC hedge} (red) and using the Least Squares Monte Carlo method (dark green).}
        \label{fig: LSMC}
    \end{figure}
    
    Note that the Nested approach described in Example \ref{ex: NMC Holder} required 63231 seconds to simulate a path of $\hat u_\pi$ corresponding to a \textit{single} realization of $(X_m, Y_m, U_{m,0}, ... U_{m,m})$ whereas the Least Squares modification took roughly 6 days to simulate the initial dataset and about 5 hours to fit the neural network. Once ready, the actual computations were conducted almost instantly.
\end{example}

\begin{remark}
    The neural network architecture in Example \ref{ex: LSMC} is not optimal. It is fairly clear from Fig.~\ref{fig: LSMC}: even though the LSMC Monte Carlo estimate preserves in general the shape of the reference optimal hedge simulated using the NMC approach, the difference between the two is still quite substantial. We emphasize that performance of the NMC method heavily depends on the chosen non-parametric regression method, and therefore a separate investigation on performance of different regression approaches is required.
\end{remark}

\begin{remark}
    Similar method can theoretically be applied for the SVV model with a rough fractional kernel. However, if one chooses the values of $\tau_{m,i}$ as in \eqref{eq: possible choice of tau}, one may need a very high dimensionality of the ``input'' variables vector (over 1000) which requires a dataset with possibly unrealistically huge number of observations. Perhaps a different choice of $\tau_{m,i}$ (e.g. as in \cite{Abi_Jaber_2019}) may improve the situation; a separate analysis on that is needed.
\end{remark}

\bibliographystyle{acm}
\bibliography{biblio.bib}

\appendix

\section{Proof of Theorem \ref{th: approximation of volatility}}\label{appendix: proof of approximation of volatility}

    Fix $t\in[0,T]$ and consider a uniform partition $0=t_0 < t_1 <...<t_N = t$, $t_n = \frac{n}{N}t$, of the interval $[0,t]$ such that the mesh $\Delta_N := \frac{t}{N}$ satisfies the condition
    \begin{equation}\label{proofeq: condition on partition}
        c_3\Delta_N < 1,
    \end{equation}
    for $c_3$ from \textbf{(Y4)}. Denote $e_{n,m} := Y(t_n) - Y_m(t_{n})$ and observe that
    \begin{align*}
        e_{N,m} &= Y(t_{N-1}) + \int_{t_{N-1}}^{t_N} b(s, Y(s))ds + Z(t_N) - Z(t_{N-1})
        \\
        &\quad - Y_m(t_{N-1}) - \int_{t_{N-1}}^{t_N} b(s, Y_m(s))ds + Z_m(t_N) - Z_m(t_{N-1})
        \\
        &= e_{N-1, m} + \int_{t_{N-1}}^{t_N} \left(b(s, Y(s)) - b(t_N, Y(t_N))\right)ds - \int_{t_{N-1}}^{t_N} \left(b(s, Y_m(s)) - b(t_N, Y_m(t_N))\right)ds
        \\
        &\quad + \left(b(t_N, Y(t_N)) - b(t_N, Y_m(t_N))\right)\Delta_N +\left( Z(t_N) - Z(t_{N-1}) \right) - \left(Z_m(t_N) - Z_m(t_{N-1})\right).
    \end{align*}
    Note that, for each $n=1,...,N$, there exists $\Theta_{n,m}$ between $Y(t_n)$ and $Y_m(t_n)$ such that
    \[
        b(t_n, Y(t_n)) - b(t_n, Y_m(t_n)) = \frac{\partial b}{\partial y}(t_n, \Theta_{n,m}) (Y(t_n) - Y_m(t_n)) = \frac{\partial b}{\partial y}(t_n, \Theta_{n,m}) e_{n, m},
    \]
    and thus
    \begin{equation}\label{eq: error estimate 1}
    \begin{aligned}
        \bigg(1 - &\frac{\partial b}{\partial y}(t_N, \Theta_{N,m}) \Delta_N\bigg)e_{N,m}= e_{N-1, m} + \int_{t_{N-1}}^{t_N} \left(b(s, Y(s)) - b(t_N, Y(t_N))\right)ds 
        \\
        &- \int_{t_{N-1}}^{t_N} \left(b(s, Y_m(s)) - b(t_N, Y_m(t_N))\right)ds +\left( Z(t_N) - Z(t_{N-1}) \right) - \left(Z_m(t_N) - Z_m(t_{N-1})\right).
    \end{aligned}
    \end{equation}
    Observe that each $1 - \frac{\partial b}{\partial y}(t_N, \Theta_{N,m}) \Delta_N > 0$ by \eqref{proofeq: condition on partition} and denote
    \[
        \zeta_{0,m} :=1, \quad \zeta_{n,m} := \prod_{k=1}^n \left(1 - \frac{\partial b}{\partial y}(t_k, \Theta_{k,m}) \Delta_N\right),
    \]
    $\widetilde e_{n,m} := \zeta_{n,m} e_{n,m}$. By multiplying the left- and right-hand sides of \eqref{eq: error estimate 1} by $\zeta_{N-1,m}$, we obtain:
    \begingroup
    \allowdisplaybreaks
    \begin{align*}
        \widetilde e_{N,m} &= \widetilde e_{N-1,m} + \zeta_{N-1,m} \int_{t_{N-1}}^{t_N} \left(b(s, Y(s)) - b(t_N, Y(t_N))\right)ds
        \\
        &\quad - \zeta_{N-1,m} \int_{t_{N-1}}^{t_N} \left(b(s, Y_m(s)) - b(t_N, Y_m(t_N))\right)ds
        \\
        &\quad +\zeta_{N-1,m} \Big(\left( Z(t_N) - Z(t_{N-1}) \right) - \left(Z_m(t_N) - Z_m(t_{N-1})\right)\Big)
        \\
        &= \widetilde e_{N-2,m} + \sum_{n=N-1}^N\zeta_{n-1,m} \int_{t_{n-1}}^{t_n} \left(b(s, Y(s)) - b(t_n, Y(t_n))\right)ds
        \\
        &\quad - \sum_{n=N-1}^N \zeta_{n-1,m} \int_{t_{n-1}}^{t_n} \left(b(s, Y_m(s)) - b(t_n, Y_m(t_n))\right)ds
        \\
        &\quad + \sum_{n=N-1}^N \zeta_{n-1,m} \left(\left( Z(t_n) - Z(t_{n-1}) \right) - \left(Z_m(t_n) - Z_m(t_{n-1})\right)\right)
        \\
        &= \sum_{n=1}^N\zeta_{n-1,m} \int_{t_{n-1}}^{t_n} \left(b(s, Y(s)) - b(t_n, Y(t_n))\right)ds  - \sum_{n=1}^N \zeta_{n-1,m} \int_{t_{n-1}}^{t_n} \left(b(s, Y_m(s)) - b(t_n, Y_m(t_n))\right)ds
        \\
        &\quad + \sum_{n=1}^N \zeta_{n-1,m} \left(\left( Z(t_n) - Z(t_{n-1}) \right) - \left(Z_m(t_n) - Z_m(t_{n-1})\right)\right),
    \end{align*}
    \endgroup
    i.e.
    \begin{align*}
        e_{N,m} &= \sum_{n=1}^N \frac{\zeta_{n-1,m}}{\zeta_{N,m}} \int_{t_{n-1}}^{t_n} \left(b(s, Y(s)) - b(t_n, Y(t_n))\right)ds - \sum_{n=1}^N \frac{\zeta_{n-1,m}}{\zeta_{N,m}} \int_{t_{n-1}}^{t_n} \left(b(s, Y_m(s)) - b(t_n, Y_m(t_n))\right)ds
        \\
        &\quad + \sum_{n=1}^N \frac{\zeta_{n-1,m}}{\zeta_{N,m}} \left(\left( Z(t_n) - Z(t_{n-1}) \right) - \left(Z_m(t_n) - Z_m(t_{n-1})\right)\right).
    \end{align*}
    Next, observe that, by \textbf{(Y4)} and \eqref{proofeq: condition on partition}, for any $n=1,...,N$
    \[
        \frac{\zeta_{n-1,m}}{\zeta_{N,m}} = \prod_{k=n}^N \left( 1 - \frac{\partial b}{\partial y}(t_k, \Theta_{k,m}) \Delta_N \right)^{-1} \le \prod_{k=n}^N \left( 1 - c_3 \Delta_N \right)^{-1} \le (1 - c_3\Delta_N)^{-N}.
    \]
    Note that $(1 - c_3\Delta_N)^{-N}$ converges as $N\to\infty$ and hence is bounded w.r.t. $N$, therefore there exists a (non-random) constant $C>0$ that does not depend on $n$, $N$ or $m$ such that $\frac{\zeta_{n-1,m}}{\zeta_{N,m}} \le C$. Using this, one can easily deduce that
    \begingroup
    \allowdisplaybreaks
    \begin{equation}\label{eq: conv. of Markovian approximations 0}
    \begin{aligned}
        |e_{N,m}| &\le \left|\sum_{n=1}^N \frac{\zeta_{n-1,m}}{\zeta_{N,m}} \int_{t_{n-1}}^{t_n} \left(b(s, Y(s)) - b(t_n, Y(t_n))\right)ds\right| 
        \\
        &\qquad + \left|\sum_{n=1}^N \frac{\zeta_{n-1,m}}{\zeta_{N,m}} \int_{t_{n-1}}^{t_n} \left(b(s, Y_m(s)) - b(t_n, Y_m(t_n))\right)ds\right|
        \\
        &\qquad + \left| \sum_{n=1}^N \frac{\zeta_{n-1,m}}{\zeta_{N,m}} \left(\left( Z(t_n) - Z(t_{n-1}) \right) - \left(Z_m(t_n) - Z_m(t_{n-1})\right)\right) \right|
        \\
        &\le C \sum_{n=1}^N \int_{t_{n-1}}^{t_n} \left|b(s, Y(s)) - b(t_n, Y(t_n))\right|ds + C\sum_{n=1}^N \int_{t_{n-1}}^{t_n} \left|b(s, Y_m(s)) - b(t_n, Y_m(t_n))\right|ds
        \\
        &\qquad + \left| \sum_{n=1}^N \frac{\zeta_{n-1,m}}{\zeta_{N,m}} \left(\left( Z(t_n) - Z(t_{n-1}) \right) - \left(Z_m(t_n) - Z_m(t_{n-1})\right)\right) \right|
    \end{aligned}
    \end{equation}
    \endgroup
    Next, fix an arbitrary $\lambda \in \left( \frac{1}{1+\gamma}, H \right)$ and observe that \eqref{eq: upper and lower bounds for sandwiched volatility, general case} and \textbf{(Y2)} imply 
    \begin{align*}
        \left|b(s_1, Y(s_1)) - b(s_2, Y(s_2))\right| & \le C (L_2 + \Lambda)^{\frac{p}{\gamma \lambda + \lambda - 1}} \left( |s_1 - s_2|^H + |Y(s_1) - Y(s_2)|\right)
        \\
        & \le C (L_2 + \Lambda)^{\frac{p}{\gamma \lambda + \lambda - 1}} \left( |s_1 - s_2|^\lambda + |Y(s_1) - Y(s_2)|\right) 
    \end{align*}
    for all $s_1, s_2 \in [0,T]$. Therefore, by  Lemma \ref{lemma: Holder continuity of approximated volatilities}, 
    \begin{equation}\label{eq: conv. of Markovian approximations summ 1}
    \begin{aligned}
        \sum_{n=1}^N &\int_{t_{n-1}}^{t_n} \left|b(s, Y(s)) - b(t_n, Y(t_n))\right|ds
        \\
        &\le C(L_2 + \Lambda)^{\frac{p}{\gamma \lambda + \lambda - 1}}  \left(\sum_{n=1}^N \int_{t_{n-1}}^{t_n} | s - t_n |^{\lambda} ds +   \sum_{n=1}^N \int_{t_{n-1}}^{t_n} |Y(s) - Y(t_n)|ds \right)
        \\
        &\le C(L_2 + \Lambda)^{\frac{p}{\gamma \lambda + \lambda - 1}}  \left(\sum_{n=1}^N \int_{t_{n-1}}^{t_n} | s - t_n |^{\lambda} ds +  \Upsilon  \sum_{n=1}^N \int_{t_{n-1}}^{t_n} | s - t_n |^{\lambda} ds \right)
        \\
        & \le C(L_2 + \Lambda)^{\frac{p}{\gamma \lambda + \lambda - 1}}(1+ \Upsilon) \Delta_N^{ \lambda },
    \end{aligned}    
    \end{equation}
    where $\Upsilon$ is from Lemma \ref{lemma: Holder continuity of approximated volatilities} and has all the moments. Moreover, \eqref{eq: upper and lower bounds for sandwiched volatility, approximations} and the same argument as above yield that
    \begin{equation}\label{eq: conv. of Markovian approximations summ 2}
    \begin{aligned}
        \sum_{n=1}^N \int_{t_{n-1}}^{t_n} \left|b(s, Y_m(s)) - b(t_n, Y_m(t_n))\right|ds \le C(L_2 + \Lambda_m)^{\frac{p}{\gamma \lambda + \lambda - 1}}(1+ \Upsilon_m) \Delta_N^{\lambda}
    \end{aligned}    
    \end{equation}
    with $\Upsilon_m$ being from Lemma \ref{lemma: Holder continuity of approximated volatilities}. Finally, using Abel summation-by-parts formula, one can write:
    \begingroup
    \allowdisplaybreaks
    \begin{equation}\label{eq: conv. of Markovian approximations summ 3}
    \begin{aligned}
        &\bigg| \sum_{n=1}^N \frac{\zeta_{n-1,m}}{\zeta_{N,m}} \Big(\left( Z(t_n) - Z(t_{n-1}) \right) - \left(Z_m(t_n) - Z_m(t_{n-1})\right)\Big) \bigg|
        \\
        & = \bigg|\frac{\zeta_{N-1,m}}{\zeta_{N,m}} (Z(t_N) - Z_m(t_N)) - \sum_{n=1}^{N-1} (Z(t_n) - Z_m(t_n)) \left(\frac{\zeta_{n,m}}{\zeta_{N,m}} - \frac{\zeta_{n-1,m}}{\zeta_{N,m}}\right)\bigg| 
        \\
        & = \bigg|\frac{\zeta_{N-1,m}}{\zeta_{N,m}} (Z(t_N) - Z_m(t_N)) + \sum_{n=1}^{N-1} \frac{\zeta_{n-1,m}}{\zeta_{N,m}} \frac{\partial b}{\partial y}(t_n, \Theta_{n,m}) (Z(t_n) - Z_m(t_n))\Delta_N\bigg|
        \\
        &\le C |Z(t_N) - Z_m(t_N)| + C \sum_{n=1}^{N-1} \left|\frac{\partial b}{\partial y} (t_n, \Theta_{n,m})\right| |Z(t_n) - Z_m(t_n)|\Delta_N.
    \end{aligned}    
    \end{equation}
    \endgroup
    Random variables $\Theta_{n,m}$ lie between $Y(t_n)$ and $Y_m(t_n)$, hence
    \begin{gather*}
        \Theta_{n,m} - \varphi(t) \ge \frac{L_1}{(L_2 + \Lambda)^\frac{1}{\gamma \lambda + \lambda - 1}} \wedge \frac{L_1}{(L_2 + \Lambda_m)^\frac{1}{\gamma \lambda + \lambda - 1}}, 
        \\
        \psi(t) - \Theta_{n,m} \ge \frac{L_1}{(L_2 + \Lambda)^\frac{1}{\gamma \lambda + \lambda - 1}} \wedge \frac{L_1}{(L_2 + \Lambda_m)^\frac{1}{\gamma \lambda + \lambda - 1}},
    \end{gather*}
    and thus it is easy to see by \textbf{(Y4)} that 
    \[
        \left| \frac{\partial b}{\partial y} (t_n, \Theta_{n,m}) \right| \le C\left(1 + (L_2 + \Lambda)^{\frac{q}{\gamma \lambda + \lambda - 1}} + (L_2 + \Lambda_m)^{\frac{q}{\gamma \lambda + \lambda - 1}}\right),
    \]
    where $C$ is, as always, a deterministic constant that does not depend on $n$, $N$ or $m$. Therefore,
    \begin{align*}
        \bigg| \sum_{n=1}^N &\frac{\zeta_{n-1,m}}{\zeta_{N,m}} \left(\left( Z(t_n) - Z(t_{n-1}) \right) - \left(Z_m(t_n) - Z_m(t_{n-1})\right)\right) \bigg|
        \\
        &\le C \left(|Z(t_N) - Z_m(t_N)| + \xi_m \sum_{n=1}^{N-1} |Z(t_n) - Z_m(t_n)|\Delta_N \right),
    \end{align*}
    where $\xi_m := 1 + (L_2 + \Lambda)^{\frac{q}{\gamma \lambda + \lambda - 1}} + (L_2 + \Lambda_m)^{\frac{q}{\gamma \lambda + \lambda - 1}}.$
   
    Summarizing \eqref{eq: conv. of Markovian approximations 0} and \eqref{eq: conv. of Markovian approximations summ 1}--\eqref{eq: conv. of Markovian approximations summ 3}, we have that
    \begin{align*}
        |Y(t) - Y_m(t)| &\le C(L_2 + \Lambda)^{\frac{p}{\gamma \lambda + \lambda - 1}}(1+ \Upsilon) \Delta_N^{\lambda } + C(L_2 + \Lambda_m)^{\frac{p}{\gamma \lambda + \lambda - 1}}(1+ \Upsilon_m) \Delta_N^{ \lambda }
        \\
        &\quad + C |Z(t) - Z_m(t)|  + C \xi_m \sum_{n=1}^{N-1} |Z(t_n) - Z_m(t_n)|\Delta_N 
    \end{align*}
    and, moving $\Delta_N \to 0$, we obtain
    \begin{align*}
        |Y(t) - Y_m(t)| &\le C \left(|Z(t) - Z_m(t)| + \xi_m \int_0^t |Z(s) - Z_m(s)| ds \right).
    \end{align*}
    It remains to notice that $\sup_{m\ge 1} \mathbb E [\xi_m^r] < \infty$ for any $r>0$ due to \eqref{eq: uniform bound on Holder seminorms of the noise} from Lemma \ref{lemma: Holder continuity of approximated noises}.


    \section{Estimates for the increments of $X$ and $X_m$}\label{appendix: estimates for price increments}

    In this Section, we will provide some technical estimates for the increments of $X$ and $X_m$, where $X$ is the SVV discounted price \eqref{eq: X} and $X_m$ is its approximation \eqref{eq: X approximated}. To allow for compact writing, we also denote
    \[
        W(t) := \rho B_1(t) + \sqrt{1 - \rho^2} B_2(t), \quad t\in [0,T],
    \]
    where $\rho \in (-1,1)$ is from \eqref{eq: S}.

    \begin{lemma}\label{lemma: bounds for conditional increments}
    Let Assumptions \ref{assum: kernel}, \ref{assum: drift} and \ref{assum: approx kernels} hold. For any pair $0 \le t_1 \le t_2 \le T$, there exist positive constants $C_1 < C_2$ that do not depend on $m$ or the particular choice of $t_1, t_2$ such that
    \[
        C_1(t_2-t_1)X^2(t_1) \le \mathbb E \left[(\Delta X)^2~|~\mathcal F_{t_1}\right] \le C_2(t_2-t_1)X^2(t_1)
    \]
    and
    \[
        C_1(t_2-t_1)X_m^2(t_1) \le \mathbb E \left[(\Delta X_m)^2~|~\mathcal F_{t_1}\right] \le C_2(t_2-t_1)X_m^2(t_1).
    \]
    Here above $\Delta X := X(t_2) - X(t_1)$, $\Delta X_m := X_m(t_2) - X_m(t_1)$.
\end{lemma}
\begin{proof}
    We only give the proof for the process $X$; the one for $X_m$ is identical. By the very definition of $X$ and It\^o lemma, we have
    \begin{equation}\label{proofeq: bounds for conditional increments 1}
    \begin{aligned}
        \mathbb E \left[(\Delta X)^2~|~\mathcal F_{t_1}\right] &= \mathbb E \left[2\int_{t_1}^{t_2} (X(s) - X(t_1)) Y(s) X(s) dW(s) ~\bigg|~\mathcal F_{t_1}\right] + \mathbb E \left[\int_{t_1}^{t_2} Y^2(s) X^2(s) ds ~\bigg|~\mathcal F_{t_1}\right]
        \\
        & = \int_{t_1}^{t_2} \mathbb E \left[Y^2(s) X^2(s) ~\big|~\mathcal F_{t_1}\right] ds.
    \end{aligned}
    \end{equation}
    We remark that the above relies on the boundedness of $Y$ and the moments of $X$.
    
    Let now $s>t_1$ be fixed. Note that, with probability 1,
    \begin{equation}\label{proofeq: bounds for conditional increments 2}
    \begin{aligned}
        \mathbb E \left[X^2(s) ~\big|~\mathcal F_{t_1}\right]\min_{t\in[0,T]} \varphi(t) \le \mathbb E \left[Y^2(s) X^2(s) ~\big|~\mathcal F_{t_1}\right] \le  \mathbb E \left[X^2(s) ~\big|~\mathcal F_{t_1}\right]\max_{t\in[0,T]} \psi(t)
    \end{aligned}
    \end{equation}
    and
    \begin{align*}
        \mathbb E \left[X^2(s) ~\big|~\mathcal F_{t_1}\right] &= \mathbb E \left[\exp\left\{ 2 \int_0^s Y(u) dW(u) - \int_0^s Y^2(u) du \right\} ~\big|~\mathcal F_{t_1}\right]
        \\
        &= \mathbb E \left[\exp\left\{  \int_0^s 2Y(u) dW(u) - \frac{1}{2}\int_0^s (2Y(u))^2 du + \int_0^s Y^2(u) du \right\} ~\big|~\mathcal F_{t_1}\right],
    \end{align*}
    whence, by the boundedness of $Y$,
    \begin{align*}
        \mathbb E \left[X^2(s) ~\big|~\mathcal F_{t_1}\right] &\le e^{T\max_{t\in[0,T] }\psi(t)}\mathbb E \left[\exp\left\{ \int_0^s 2Y(u) dW(u) - \frac{1}{2}\int_0^s (2Y(u))^2 du \right\} ~\big|~\mathcal F_{t_1}\right]
        \\
        \mathbb E \left[X^2(s) ~\big|~\mathcal F_{t_1}\right] &\ge \mathbb E \left[\exp\left\{ \int_0^s 2Y(u) dW(u) - \frac{1}{2}\int_0^s (2Y(u))^2 du \right\} ~\big|~\mathcal F_{t_1}\right].
    \end{align*}
    Next, by the Novikov criterion, the process
    \[
        \exp\left\{\int_0^s 2Y(u) dW(u) - \frac{1}{2}\int_0^s (2Y(u))^2 du \right\}, \quad s\in[0,T],
    \]
    is a martingale, and hence, for $s>t_1$,
    \begin{align*}
        \mathbb E \left[\exp\left\{  \int_0^s 2Y(u) dW(u) - \frac{1}{2}\int_0^s (2Y(u))^2 du  \right\} ~\big|~\mathcal F_{t_1}\right] &= \exp\left\{  \int_0^{t_1} 2Y(u) dW(u) - \frac{1}{2}\int_0^{t_1} (2Y(u))^2 du \right\}
        \\
        &= X^2(t_1) \exp\left\{ -\int_0^{t_1} Y^2(s)ds \right\}.
    \end{align*}
    Therefore,
    \begin{align*}
        \mathbb E \left[\exp\left\{  \int_0^s 2Y(u) dW(u) - \frac{1}{2}\int_0^s (2Y(u))^2 du  \right\} ~\big|~\mathcal F_{t_1}\right] &\le X^2(t_1),
        \\
        \mathbb E \left[\exp\left\{  \int_0^s 2Y(u) dW(u) - \frac{1}{2}\int_0^s (2Y(u))^2 du  \right\} ~\big|~\mathcal F_{t_1}\right] &\ge X^2(t_1) e^{ -T\max_{t\in[0,T]} \psi(t) },
    \end{align*}
    and we can now write
    \begin{equation}\label{proofeq: bounds for conditional increments 3}
        e^{ -T\max_{t\in[0,T]} \psi(t) }X^2(t_1) \le \mathbb E \left[X^2(s) ~\big|~\mathcal F_{t_1}\right] \le e^{T\max_{t\in[0,T] }\psi(t)} X^2(t_1).
    \end{equation}
   
    Finally, by \eqref{proofeq: bounds for conditional increments 1}--\eqref{proofeq: bounds for conditional increments 3}, 
    \[
        C_1(t_2 - t_1)  X^2(t_1) \le \mathbb E \left[(\Delta X)^2~|~\mathcal F_{t_1}\right] \le C_2(t_2 - t_1) X^2(t_1)
    \]
    where
    \begin{align*}
        C_1  := \min_{t\in[0,T]} \varphi(t) e^{-T\max_{t\in[0,T] }\psi(t)}, \quad  C_2 := \max_{t\in[0,T]} \psi(t) e^{T\max_{t\in[0,T] }\psi(t)}.
    \end{align*}
\end{proof}

\begin{lemma}\label{lemma: difference between increments}
    Let Assumptions \ref{assum: kernel}, \ref{assum: drift} and \ref{assum: approx kernels} hold and both $\mathcal K$ and $\mathcal K_m$, $m\ge 1$, have the form
    \[
        \mathcal K(t,s) = \mathcal K(t-s)\mathbbm 1_{s<t}, \quad \mathcal K_m(t,s) = \mathcal K_m(t-s)\mathbbm 1_{s<t}, \quad t,s \in [0,T].
    \]
    For any $0 \le t_1 \le t_2 \le T$ and $r\ge 2$, there exists a constant $C>0$ that does not depend on $m$ or the particular choice of $t_1, t_2 \in [0,T]$ such that
    \[
        \mathbb E[|\Delta X - \Delta X_m|^r] \le C (t_2 - t_1)^{\frac{r}{2}} \lVert \mathcal K - \mathcal K_m \rVert^r_{L^2([0,T])}
    \]
    with $\Delta X := X(t_2) - X(t_1)$, $\Delta X_m := X_m(t_2) - X_m(t_1)$.
\end{lemma}
\begin{proof}
    Using uniform boundedness of $Y$ and $Y_m$, the Burkholder-Davis-Gundy and Jensen inequalities as well as Theorem \ref{th: approximation of volatility}, one can write:
    \begin{align*}
        \mathbb E& \left[ |\Delta X - \Delta X_m|^r \right] = \mathbb E \left[ \left|\int_{t_1}^{t_{2}} (Y(s)X(s) - Y_m(s)X_m(s))dW(s)\right|^r \right]
        \\
        &\le C \mathbb E \left[ \left|\int_{t_1}^{t_2} (Y(s)X(s) - Y_m(s)X_m(s))^2 ds\right|^{\frac{r}{2}} \right]
        \\
        &\le C(t_2 - t_1)^{\frac{r}{2}-1}  \int_{t_1}^{t_2} \mathbb E \left[|Y(s)X(s) - Y_m(s)X_m(s)|^r\right] ds 
        \\
        &\le C(t_2 - t_1)^{\frac{r}{2}-1} \int_{t_1}^{t_2} \mathbb E\left[|X(s) - X_m(s)|^r\right] ds + C(t_2 - t_1)^{\frac{r}{2}-1} \int_{t_1}^{t_2} \mathbb E\left[ X^r(s)|Y(s) - Y_m(s)|^r\right]ds
        \\
        &\le C(t_2 - t_1)^{\frac{r}{2}-1}\int_{t_1}^{t_2} \mathbb E\left[|X(s) - X_m(s)|^r\right] ds + C(t_2 - t_1)^{\frac{r}{2}-1}\int_{t_1}^{t_2} \mathbb E\left[ X^r(s)|Z(s) - Z_m(s)|^r\right]ds 
        \\
        &\quad + C(t_2 - t_1)^{\frac{r}{2}-1}\int_{t_1}^{t_2} \mathbb E\left[ X^r(s) \xi^{r}_m \int_0^s |Z(u) - Z_m(u)|^r du \right]ds,
    \end{align*}
    where the random variable $\xi_m$ does not depend on $t_1$ or $t_2$.
    Now, since
    \[
         \mathbb E\left[\sup_{s\in [0,T]}|X(s) - X_m(s) |^r\right] \le C \lVert \mathcal K - \mathcal K_m \rVert^r_{L^2([0,T])}  
    \]
    by Theorem \ref{th: approx with difference kernel}, it is clear that
    \[
        \int_{t_1}^{t_2} \mathbb E\left[|X(s) - X_m(s)|^r\right] ds \le C(t_2 - t_1) \lVert \mathcal K - \mathcal K_m \rVert^r_{L^2([0,T])}.
    \]
    Next, \eqref{eq: sup of S and X} in Theorem \ref{th: properties of price} implies that there exists a constant $C$ that does not depend on $t_1$ or $t_2$ such that
    \[
        \int_{t_1}^{t_2} \mathbb E\left[ X^{2r}(s)\right]ds < C,
    \]
    which yields
    \begin{align*}
        \int_{t_1}^{t_2} & \mathbb E\left[ X^r(s)|Z(s) - Z_m(s)|^r\right]ds \le \left( \int_{t_1}^{t_2} \mathbb E\left[ X^{2r}(s)\right]ds \right)^{\frac{1}{2}} \left(\int_{t_1}^{t_2} \mathbb E\left[|Z(s) - Z_m(s)|^{2r}\right]ds \right)^{\frac{1}{2}}
        \\
        & \le C \left(\int_{t_1}^{t_2} \mathbb E\left[|Z(s) - Z_m(s)|^{2r}\right]ds \right)^{\frac{1}{2}} \le C(t_2 - t_1) \lVert \mathcal K - \mathcal K_m \rVert^r_{L^2([0,T])}.
    \end{align*}
    Similarly, by \eqref{eq: sup of S and X} and \eqref{eq: xim uniformly bounded moments} we can write
    \begin{align*}
        \int_{t_1}^{t_2} \mathbb E\left[ X^r(s) \xi^{r}_m \int_0^s |Z(u) - Z_m(u)|^r du \right]ds &\le C \int_{t_1}^{t_2} \int_0^s \left(\mathbb E\left[|Z(u) - Z_m(u)|^{2r}\right] \right)^{\frac{1}{2}} duds
        \\
        &\le C (t_2 - t_1) \lVert \mathcal K - \mathcal K_m \rVert^r_{L^2([0,T])},
    \end{align*}
    which finalizes the proof.
\end{proof}
    
\end{document}